\documentclass[12pt,a4,english,leqno]{article}
\usepackage[latin9]{inputenc}
\usepackage{geometry}
\usepackage{verbatim}
\usepackage{amsmath}
\usepackage{amsfonts}
\usepackage{amsthm}
\usepackage{mathtools}
\usepackage{graphicx}
\usepackage{longtable}
\usepackage{setspace}
\usepackage{upgreek}
\usepackage{comment}
\usepackage[round,sort,authoryear]{natbib}
\usepackage[colorlinks=true,linkcolor=blue, citecolor=blue, urlcolor = blue]{hyperref}

\usepackage{titlesec}
\usepackage{babel}
\usepackage[toc,page]{appendix}
\usepackage{booktabs,caption}

\usepackage[labelfont=bf,textfont=md]{caption}
\usepackage[labelsep=space]{caption}

\usepackage{lmodern}
\usepackage{avant}
\usepackage[T1]{fontenc}

\DeclarePairedDelimiter{\floor}{\lfloor}{\rfloor}

\usepackage{bibentry}
\nobibliography*

\usepackage{fancyhdr}
\numberwithin{equation}{section}

\titleformat*{\section}{\large \bfseries}
\titleformat*{\subsection}{\normalsize \bfseries}
\titleformat*{\subsubsection}{\small \bfseries}

\newif\ifshow 
\showfalse 

\ifshow
  \includecomment{wrap}
\else
  \excludecomment{wrap} 
  
\fi

\theoremstyle{definition}
\newtheorem{theorem}{Theorem}
\newtheorem{definition}{Definition}
\newtheorem{assumption}{Assumption}
\newtheorem{lemma}{Lemma}
\newtheorem{example}{Example}

\newtheorem{remark}{Remark}
\newtheorem{corollary}{Corollary}

\begin{document}
\pagenumbering{roman}

\title{ {\Large \textbf{Break-Point Date Estimation for Nonstationary Autoregressive and Predictive Regression Models}\thanks{\textbf{Article history:} December 2020, December 2021, August 2023. I am grateful to Jose Olmo, Jean-Yves Pitarakis and Tassos Magdalinos at the Department of Economics, University of Southampton for prior helpful discussions as well as Giuseppe Cavaliere and Sebastian Kripfganz at the University of Exeter Business School. 
\\

Lecturer in Economics, Department of Economics, University of Exeter Business School, Exeter  EX4 4PU, United Kingdom. \textit{E-mail Address}: \textcolor{blue}{christiskatsouris@gmail.com}.\\ } 
}
}

\author{ \textbf{Christis Katsouris} \\ \textit{University of Southampton and University of Exeter}\\ \\ \textcolor{blue}{This version: Working Paper}}

\date{\today}

\maketitle

\begin{abstract}
\vspace*{-0.60 em}
In this article, we study the statistical properties and asymptotic behaviour of break-point estimators in nonstationary autoregressive and predictive regression models for testing the presence of a single structural break at an unknown location in the full sample. Moreover, we investigate aspects such as how the persistence properties of covariates and the location of the break-point affects the limiting distribution of the proposed break-point estimators. 
\\

\textbf{Keywords:} nonstationary processes, persistence, local-to-unit root, structural break, break-point estimator, break magnitude, single break, multiple breaks, convergence rates, asymptotic distribution. 
\\

JEL Classification: C12, C22
\end{abstract}

\newpage 

\setcounter{page}{1}
\pagenumbering{arabic}

\newpage

\section{Introduction}

Structural break inference is an important task when the robustness of econometric estimation is concerned. Although various existing methodologies in the time series econometrics literature focus on testing for the presence of parameter instability in regression coefficients, the estimation of the exact break-points and their statistical properties is more challenging, especially under the assumption of regressors nonstationarity. We develop point estimation and asymptotic theory for break-point estimators in nonstationary autoregressive and predictive regression model with nonstationary regressors. We first present the structural break test for the model coefficient of the predictive regression model based on the usual least squares estimate of the coefficient of the first-order autoregressive model as well as the predictive regression model. Second we using the IVX instrumentation which is found to be robust to the abstract degree of persistence (see, \cite{kostakis2015robust}), we construct IVX-based structural break tests\footnote{Relevant asymptotic theory analysis for the Wald-type statistics we employ in this article are presented by \cite{katsouris2023predictability,  katsouris2023testing}.} (see, \cite{katsouris2023predictability,katsouris2023structural,  katsouris2023testing}) and corresponding  break-point estimators. Overall, the accurate dating of a structural break is of paramount important especially since predictive regression models are commonly used to detect the so-called "pockets of predictability". Thus, correctly identifying periods of predictability regardless of the presence of a structural break is crucial for asset pricing and risk management purposes.

Determining the asymptotic distribution of the break-point estimator in nonstationary regression models is a challenging task. Specifically, for the case of a shift with fixed magnitude it can be shown that the limiting distribution of the change-point estimator depends on the underlying distribution of the innovation in a complicated manner (see, \cite{hinkley1970inference}). The least squares estimation of the change-point in mean and variance in a linear time series regression model is obtained by \cite{pitarakis2004least}, although the nonstationary properties of regressors are not considered. Furthermore, \cite{pitarakis2012jointly, pitarakis2014joint} develops suitable econometric frameworks for jointly testing the null hypothesis of no structural change in nonstationary times series regression models. More recently   \cite{dalla2020asymptotic} and   \cite{stark2022testing} consider relevant aspects for structural break testing and dating in regression models under dependence. The commonly used  assumption is that the break occurs at an unknown time location within the full sample, such that, $t = \floor{ \tau T }$, where $\tau \in (0,1)$ and these limits are also used when obtaining moment functionals. 

Specifically, predictive regressions with regressors generated via the local-to-unity parametrization has recently seen a growing attention in the literature (see, \cite{phillipsmagdal2009econometric}, \cite{gonzalo2012regime},  \cite{phillips2014confidence}, \cite{kostakis2015robust}, \cite{kasparis2015nonparametric}, \cite{demetrescu2020testing} and \cite{duffy2021estimation}). The main feature of these time series regression models, is that an $I(0)$ integrated dependent variable (such as stock returns) is regressed against a persistent predictor (such as the divided-price ratio) and this allows to construct predictability tests (see, also \cite{zhu2014predictive}). In statistical terms,  the process $\left( Y_t, \boldsymbol{X}_{t-1} \right)_{ t \in \mathbb{Z} }$ with a martingale difference sequence $\xi_t = \left(    u_t, v_t \right)^{\prime}$ such that $\left( Y_t, \boldsymbol{X}_{t-1}, \xi_t \right)_{ t \in \mathbb{Z} }$ is generated by a predictive regression model and $\boldsymbol{X}_{t}$ is an autoregressive process expressed using the local-to-unity parametrization, it allows us to examine the persistence properties of regressors across various regimes.  On the other hand, these frameworks operate under the assumption of a parameter constancy in the full sample.

In this paper, we study the break-point estimators for structural break tests in  predictive regression models with possible nonstationary regressors. The stability of the autoregressive processes is determined by the local-to-unity parametrization. Specifically, we focus on autoregressive processes which are close to the unit boundary but have different order of convergence, namely high persistent regressors which are $o_P \left(  n^{-1 / 2} \right)$ and mildly integrated regressors which are  $o_P \left(  n^{- \upgamma / 2} \right)$ for some $\gamma \in (0,1)$ and a positive persistence coefficient $c > 0$. These features do complicate the asymptotic theory analysis in some extend but their properties are useful for break-point estimation and dating in the aforementioned settings. For instance, conventional structural break tests for the parameters of linear regression models employ the widely used sup-Wald test proposed by \cite{andrews1993tests}. However, the distributional theory of the Andrews's test crucially depends on the strict stationarity assumption of regressors. In contrast to the  literature that focuses on structural break testing in linear regressions, the predictive regression model is usually fitted to economic datasets which contain time series that are highly persistent. Therefore, within such econometric environment the traditional law of large numbers and central limit theorems can invalidate the standard econometric assumptions of linear regression models, which affect the large sample approximations. As a result distorted inferences can occur when testing for parameter instability in predictive regressions when these features are not accommodated in the asymptotic theory of the tests. 

Our first objective is to theoretically demonstrate the impact of the presence of the nuisance parameter of persistence to the limiting distribution of test statistics and break-point estimators. Specifically, the limit result obtained by \cite{andrews1993tests} implies the use of the supremum functional on the Brownian Bridge process defined as $\mathsf{sup}_{ s \in [0,1] } \big[ W_n(s) - s W_n (1) \big]$. Under weak convergence, we have process convergence where $B$ is a Brownian motion, that is, a pivot process, hence enabling the practitioner to use standard tabulated critical values. Conveniently, \cite{katsouris2023predictability,  katsouris2023testing} show that the OLS based sup-Wald statistic in the presence of nearly integrated regressors weakly converges into the standard NBB limit, however in the case of regressors with high persistence the same limit is no longer valid. As a result, this complicates the asymptotic theory analysis of break-point estimators especially since usually these nonstationary features are not a prior known. However, before proceeding to the limit behaviour of break-point estimators we discuss an alternative estimation procedure which is found to be robust to the unknown persistence properties can be applied to a structural-break testing while producing similar asymptotic behaviour. In this direction, one can employ the IVX instrumental variable estimation approach proposed by \cite{phillipsmagdal2009econometric} and  extensively examined by \cite{kostakis2015robust} in the context of predictability tests and by \cite{gonzalo2012regime} in the context of predictability tests for threshold predictive regression models\footnote{More specifically, the limit theory of \cite{kostakis2015robust} provides a unified framework for robust inference and testing regardless the persistence properties of regressors. A simple example is the application of the IVX-Wald test for inferring the individual statistical significance of predictors under abstract degree of persistence. Further scenarios such as predictors of mixed integration order see \cite{phillips2013predictive} and \cite{phillips2016robust} in which cases the mixed normality assumption still holds.}. 

\newpage  

Thus, we study the statistical inference problem of structural break testing at an unknown break-point therefore the supremum functional is implemented and two test statistics are considered based on two different parameter estimation methods. The first estimation method considers the OLS estimator, while the second method considers an instrumental variable based estimator, namely the IVX estimator proposed by \cite{phillipsmagdal2009econometric}. These two estimators of the predictive regression model have different finite-sample and asymptotic properties, which allows us to compare the limiting distributions of the proposed tests for the two different types of persistence of the regressors. For both test statistics we assume that the regressors included in the model are permitted to be only one of the two persistence types which simplifies the asymptotic theory of the tests, however the presence of nuisance parameters under the null of parameter constancy, that is, the unknown break-fraction and the coefficient of persistence, requires careful examination of the asymptotic theory. An additional caveat is the inclusion of an intercept in the predictive regression model which induces different limiting distributions when is assumed to be stable vis-a-vis the case in which is permitted to shift. In summary, we are interested to examine the consistency and convergence rates of the break-point estimators that correspond to these two estimation methodologies for the coefficients of the predictive regression model.  

The asymptotic theory of the present paper hold due to the  invariance principle of the partial sum process of $x_t$, where $x_t = \left( 1 - \frac{c}{n^{\upgamma}} \right) x_{t - 1}$, as proposed by \cite{phillips1987time} and is considered to be the building block for related limit results when considering time series models. We denote with $\hat{ \mathcal{U} }_n(s) := \frac{1}{ \sqrt{n}  } \sum_{ t = 1 }^{ \floor{ ns } } x_{t}$, for some $s \in [0,1]$ and with $\hat{ \mathcal{U} }_{ n^{\upgamma} }(s) := \frac{1}{ T^{ \upgamma / 2 }  } \sum_{ t = 1 }^{ \floor{ n^{ \upgamma } s } } x_{t}$, for some $s \in [0,1]$ and $0 < \upgamma  < 1$ for the invariance principle of the partial sum process of $x_t$ in the case of mildly integrated processes, for the corresponding limit as proposed by \cite{phillips2005limit}. Motivated from the aforementioned seminal work, as well as the framework of \cite{phillipsmagdal2009econometric}, similarly in this paper we use these invariance principles of the partial sums processes that correspond to the instrumental variable, IVX, proposed by \cite{phillipsmagdal2009econometric}, specifically within a structural break testing framework. These results allow us to formally obtain the limiting distributions of the proposed tests with respect to the nuisance parameter of persistence along with the unknown break-point location, and observe in which cases we obtain nuisance-free inference that can simplify significantly the hypothesis testing procedure.

All random elements are defined with a  probability space, denoted by $\left( \Omega, \mathcal{F}, \mathbb{P} \right)$. All limits are taken as $n \to \infty$, where $n$ is the sample size. The symbol $"\Rightarrow"$ is used to denote the weak convergence of the associated probability measures as $n \to \infty$. The symbol $\overset{d}{\to}$ denotes convergence in distribution and $\overset{\text{plim}}{\to}$ denotes convergence in probability, within the probability space. Let $\left\{ Y_t, \boldsymbol{X}_t \right\}_{t=1}^n$, where $\boldsymbol{X}_t = \left( X_{1t},..., X_{pt}   \right)$ denote the corresponding random variables of the underline joint distribution function. The rest of the paper is organized as follows. Section \ref{Section2} presents the asymptotic theory for break-point detection in AR(1) models. Section \ref{Section3} considers the asymptotic theory for break-point detection predictive regressions. 

\newpage

\subsection{Contributions to the literature}

Our study contributes to the time series econometrics literature in several ways. Firstly, we propose a set of Wald type statistics for detecting parameter instability in predictive regression models robust to the persistence properties of regressors. We derive analytical forms of the limit distributions of these test statistics and show under which conditions these limit results correspond to the conventional NBB result providing this way a clear estimation and inference strategy to practitioners interested to implement a structural break test in predictive regression models with possibly nonstationary regressors. Therefore the significance of the proposed framework in the broader econometric literature includes the provision of detailed asymptotic results for these structural break tests as well as  necessary data transformations and functional forms which are can be implemented so that statistical inference can be simplified.  Therefore, the construction of structural break tests in predictive regression models would not have been possible without one first consider the asymptotic behaviour of estimators of the predictive regression model. 

Secondly, in a similar spirit as in \cite{saikkonen2006break}, we aim to investigate the properties of estimators of the time period where a shift has taken place. In particular, under the assumption of a possible single structural break, we point identify the shift in the model coefficients of the first order autoregressive and predictive regression model. Moreover, two alternative estimators for the break date are considered, and their asymptotic properties are derived under various assumptions regarding the local alternatives in which case the size of the shift is considered. These results have various further applications as they can  then be used to explore the implications of inference in predictive regression models after estimation of breaks. Lastly, we aim to perform a more detailed and more insightful investigation of the small-sample properties of the break date estimators and the resulting structural break tests by extending the simulation design and empirical findings of \cite{katsouris2023predictability}. Notice that the break-point date estimator denoted with $k = \uptau n$, where $\uptau \in (0,1)$, needs to be estimated using a criterion function. However,  the estimation approach of the econometric model will also consequently affect the statistical properties of the corresponding break-point estimator. Rearranging the break-point estimators gives the following expression 
\begin{align*}
\hat{\uptau} = \frac{ \hat{k} }{n}, \ \ \ \text{where} \ \ \ \left( \hat{\uptau} - \uptau \right) = o_p(1).
\end{align*}  
A relevant asymptotic theory question of interest is under which conditions the above convergence in probability holds. Roughly speaking, the main idea of the framework here is that asymptotically the break date can then be located at the true break date or within a neighbourhood of the true break date. Therefore, one will need to carefully consider the required assumptions that provide identification conditions for the break date. In addition, we will need to consider possible aspects of the structural break testing environment under regressors nonstationarity that can potentially make the break date estimator $\hat{\uptau}$ inconsistent, resulting to an incorrect estimation of the break date for the modelling environment under consideration.

\newpage

In terms of the stochastic integral approximations we consider in this paper, a relevant literature include \cite{phillips1987time, phillips1987towards, phillips1988regression} as well as \cite{kurtz1991weak} and \cite{hansen1992convergence} who present various examples of weak convergence of stochastic integrals and stochastic differential equations upon which the limit theory of this paper is based on. The related limit theory proves that the continuous time OU diffusion process given by
\begin{align}
d J (t) = c J (t) dt + \sigma dW (t), \ J(0) = b, t > 0 
\end{align}
where $c$ and $\sigma >0$ are unknown parameters and $W_t$ is the standard Wiener process, has a unique solution to $\{ J(t) \}_{ t \in [0,1]}$, such as $J(t) = \displaystyle \mathsf{exp} \left( c t  \right) b + \sigma \int_0^t \mathsf{exp} \left\{ c( t - s)  \right\} dW (s) \equiv \text{exp} \left(c t \right) + \sigma J_{c} \left( t \right)$ (see, \citep{perron1991continuous}). This representation provides a way of determining the asymptotic convergence of each of the component which the expression of the Wald type statistic can be decomposed to.  
\color{black}

\subsection{Related Literature}

When the break-point is known one can apply a Chow-type test statistic. In particular, \cite{sun2022asymptotically} investigate the asymptotic distribution of a Chow-based test in the presence of heteroscedasticity and autocorrelation. However, when the break-point is unknown, the estimation procedure for break-point date estimation is usually based on the optimization of a criterion function. Moreover, the presence of a possible single against multiple-breaks will require modifications to the formulation of the relevant hypotheses, the criterion function as well as the asymptotic theory analysis. Specifically, in the literature of structural break testing for linear regression models, \cite{bai1997estimating} propose a statistical procedure for detecting multiple breaks sequentially (one-by-one) rather than using a simultaneous estimation approach of multiple breaks. In summary while  \cite{bai1997estimating} proposed a sample-splitting method to estimate the breaks one at a time by minimizing the residual sum of squares, \cite{bai1998estimating} proposed to estimate the breaks simultaneously by minimizing the residual sum of squares. The advantages of the former method lies in its computational savings and its robustness to misspecification in the number of breaks. A number of issues arise in the presence of multiple breaks. More precisely,  the determination of the number of breaks, the estimation of the break points given the number as well as the statistical analysis of the resulting estimators. Simultaneous and sequential methods are fundamentally different methodologies that yield different break-point estimators. However, one of the drawbacks of sequential break-point algorithms is the complexity in deriving convergence rates of the estimated break-point. In particular, \cite{bai1997estimating} demonstrate that sequentially obtaining estimated break points are $n-$consistent, which corresponds to the same rate as in the case of simultaneous estimation\footnote{
In particular, for the simultaneous estimators of break, their asymptotic distributions in stationary models are symmetric, but the computational burden is heavy. The least-squares operations are of order $O( n^2 )$ even under the most efficient algorithm. In contrast, for the sequential estimators of breaks, the computational burden is light (the least-squares operations are of order $O( n )$, but the asymptotic distributions of the estimators, are asymmetric due to the misspecification of the models. Hence, additional efforts, such as repartitioning the sample, are needed in order to obtain symmetrical asymptotic distributions of the estimators.}.

\newpage 

Although these features have not been investigated in the case of predictive regression models with possible nonstatationary regressors or regressors with mixed integration order. In this article, we aim to provide some insights on some relevant cases (not all of the nonstationary regimes). Within our framework, the first stage of the procedure implies testing for the possible presence of a structural break under high persistence, that is, the regressors of the autoregressive and the predictive regression models are parametrized using the local-to-unity parameter which induces a nearly-integrated process. Furthermore, the accuracy of the estimator depends on whether the break-point estimator is bounded. Thus, the robustness of the break-point estimator can be verified by its ability to be as close as possible to the true break-point within the full sample. 
\color{black}

From the unit root and structural break  literature perspective, relevant testing methodologies include \cite{saikkonen2006break}  as well as the single-break homoscedasticity-based persistence persistence change tests proposed by \cite{harvey2006modified}. Testing for multiple break-points are discussed in the studies of \cite{lumsdaine1997multiple} and \cite{carrion2009gls}. Moreover,  \cite{kejriwal2008limit} study estimation and inference in cointegrated regression models with multiple structural changes allowing both stationary and integrated regressors; by deriving the constistency, rate of convergence and the limit distribution of the estimated break fractions. If the coefficients of the integrated regressors are allowed to change, the estimated break fractions are asymptotically dependent so that confidence itervals need to be constructed jointly. Recently, \cite{kejriwal2020bootstrap} propose bootstrap procedures for detecting multiple persistence shifts in heteroscedastic time series such as cointegrating regressions. The bootstrap procedure proposed by \cite{kejriwal2020bootstrap} for detecting multiple breaks as an example here is as below: 
\begin{itemize}

\item[Step 1.] If the null hypothesis is not rejected at the desired level of significance, stop the procedure and conclude there is no evidence of instability. Otherwise, obtain the break date estimate $\hat{\lambda}$ by minimizing the sum of squared residuals and proceed to the following step. 

\item[Step 2.] Conduct an F-test using chi-squared critical values for the equality of the coefficients across regimes on the subset of coefficients of interest allowing the others to change at the estimated breakpoint. Upon a rejection, conclude in favor of a structural change in the subvector of interest, otherwise the stability cannot be rejected. 

\end{itemize}
The asymptotic validity of the two-step procedure follows from \textit{(i)} the test in the first step is asymptotically pivotal under the null and consistent against alternatives involving a change in at least one parameter and  \textit{(ii)} the break fraction is consistently estimated as long as any of the parameters are subject to a break. In particular, the second fact ensures that the F-test in the second step converges to a chi-square distribution under the null hypothesis of no structural change in the subvector of interest. More precisely, this result follows since the estimate of the break fraction is fast enough to ensure that the limiting distribution of the parameter estimate is the same that would prevail if the break date was known. In summary, the particular framework reveals that existing partial break sup-Wald tests diverge with $n$ when the coefficients are not being tested are subject to change. Thus, \cite{kejriwal2020bootstrap}  propose a simple two-step procedure which first tests for joint parameter stability and subsequently conducts a standard chi-squared stability test on the coefficients of interest allowing the other coefficients to change at the breakpoints estimated by minimizing the sum of squared residuals in the pure structural change model. The procedure proposed by \cite{kejriwal2020bootstrap} estimates the number of breaks using a sequential test of the null hypothesis of $( \ell \geq 1 )$ breaks against the alternative of $( \ell + 1 )$ breaks. The particular approach is useful especially in cases when the null of no break could be rejected against the alternative hypothesis of at least one break. Lastly, \cite{casini2022generalized} propose a generalized laplace inference in multiple change-points models framework, although not suitable for detecting multiple-breaks in cointegrating and predictive regression models.  

Generally speaking, estimation and inference procedures of econometric models that do not account for such data-driven selection of nuisance parameters such as an unknown structural break or an unknown threshold variable can perform poorly when used for empirical studies \citep{andrews2021inference} (see, also \cite{gonzalo2012regime,gonzalo2017inferring}). In this direction, \cite{zhu2022testing} propose a framework which considers the possible presence of both a structural break and a threshold effect in predictive regression models, robust to the degree of persistence. In the self-normalization literature two relevant studies include \cite{choi2022subsample} and \cite{zhang2018unsupervised}. Then, to correct for serial dependence, it often requires a consistent estimator of the long-run variance, namely, the spectral density at zero frequency. The particular quantity involves autocovariances of all orders, and a data-driven bandwidth is usually needed for its estimator to be adaptive to the underlying dependence. Specifically, \cite{shao2010testing} uses the self-normalization approach for single change-point testing in time series. Note that, the self-normalization approach implies that instead of restoring to a consistent estimator of the long-run variance, one relies on a sequence of recursive estimators to form the normalizer and in turn pivotalize the asymptotic distribution of the test statistic. A framework for single structural break testing based on the self-normalization approach is given by \cite{ling2007testing}. 
Lastly, \cite{andrews2021inference} study methodologies for inference after estimation of structural breaks (see, also \cite{fiteni2002robust} and \cite{busetti2003variance}).

\subsection{Illustrative Examples}

We present some relevant illustrative examples to the modelling environment under consideration. 

\begin{example}
Consider the following predictive regression model 
\begin{align}
y_t = ( \alpha_1 + \boldsymbol{\beta}_1 x_{t-1} ) \boldsymbol{1} \left\{ t \leq k  \right\} +  ( \alpha_2 + \boldsymbol{\beta}_2 x_{t-1} ) \boldsymbol{1} \left\{ t > k  \right\}  + u_t  
\end{align}
\end{example}
Consider testing the joint null hypothesis\footnote{Note that we can also test predictability in the pre-break and post-break subsamples, using $H_0: \boldsymbol{\beta}_1  = 0$ and $H_0: \boldsymbol{\beta}_1  = 0$, respectively.} $H_0: \boldsymbol{\beta}_1 = \boldsymbol{\beta}_2 = \boldsymbol{0}$. Under the null hypothesis, the predictive regression model reduces to a change in mean model as below (see, \cite{katsouris2022partial}): 
\begin{align}
y_t = \alpha_1 \boldsymbol{1} \left\{ t \leq k  \right\} +  \alpha_2  \boldsymbol{1} \left\{ t > k  \right\} + u_t     
\end{align}

\newpage

\begin{assumption}
The following conditions hold: 
\begin{itemize}

\item[\textit{(i)}.] The magnitude of the level shift can be expressed as below $\left| \alpha_2 - \alpha_1 \right| =  n^{-1 / 2} \delta_n$  for a sequence $\delta_n$ which is a function of the sample size $n$ such that $\delta_n = \mathcal{O} ( n^s )$, for some $s \in (0, 1/2]$.

\item[\textit{(ii)}.] Suppose that $k = \floor{   \uptau n }$, where $\uptau \in \Pi := ( \pi_0 , 1 - \pi_0  )$ for some $\pi_0 \in (0, 1/4)$. 
    
\end{itemize}
Under the above assumption (identification assumption) the magnitude of the level shift either is independent of the sample size or shrinks to zero at a rate slower than $n^{ -1/2 }$. The break-fraction can be consistently estimated regardless of whether $x_t$ is either stationary or nearly integrated under the null hypothesis (see, also \cite{kejriwal2008limit}). 
\end{assumption}

\begin{example}[see,\cite{casini2022generalized}]
Some relevant results include: For $r \in [0,1]$, it holds
\begin{align}
\frac{1}{ \sqrt{ n_b^0 } } \sum_{t=1}^{ \floor{ r n_b^0 }  }  z_t u_t \Rightarrow \mathcal{G}_1 (r)  
\ \ \ \ 
\frac{1}{ \sqrt{ \left( n - n_b^0 \right) } } \sum_{t= n_b^0 + 1}^{ n_b^0 + \floor{ r \left( n - n_b^0 \right) }  }  z_t u_t \Rightarrow \mathcal{G}_2 (r)  
\end{align}
where $\mathcal{G}_i ( \cdot )$ is a multivariate Gaussian process on $[0,1]$ with zero mean and covariance such that $\mathbb{E} \big[ \mathcal{G}_i ( \mathsf{u} ) \mathcal{G}_i ( \mathsf{s}  )  \big] = \mathsf{min} \left\{ \mathsf{u}, \mathsf{s}   \right\} \times \Sigma_i$, where $i \in \left\{ 1,2 \right\}$.   
\begin{align} 
\Sigma_1 = \underset{ n \to \infty }{ \mathsf{lim}  } \ \mathbb{E} \left[ \frac{1}{ \sqrt{ n_b^0 } } \sum_{t=1}^{ \floor{ r n_b^0 }  }  z_t u_t  \right]^2
\ \ \ \ 
\Sigma_2 = \underset{ n \to \infty }{ \mathsf{lim}  } \ \mathbb{E} \left[  \frac{1}{ \sqrt{ \left( n - n_b^0 \right) } } \sum_{t= n_b^0 + 1}^{ n_b^0 + \floor{ r \left( n - n_b^0 \right) }  }  z_t u_t   \right]^2
\end{align}
For any $0 < r_0 < 1$, $r_0 < \uptau_0$, $\displaystyle \frac{1}{n} \sum_{ \floor{ r_0 n } + 1 }^{  \floor{  \uptau_0 n  }  }  z_t z_t^{\prime} \overset{p}{\to}  ( \uptau_0 - r_0 ) V_1$, 
and $\uptau_0 < r_0$, $\displaystyle \frac{1}{n} \sum_{ \floor{ \uptau_0 n } + 1 }^{  \floor{  r_0 n  }  }  z_t z_t^{\prime} \overset{p}{\to}  ( r_0 - \uptau_0 ) V_2$.
\end{example}
A suitable stopping rule used for the change-point detection rely either on thresholding or on the optimization
of a model selection criterion. Various methods for multiple change-point detection exist in the literature which include the dynamic programming method, to detect multiple change points in the exponential family of distributions. 
According to \cite{casini2022generalized}, for multiple break-points such that $m$ change-points, the inference framework can be constructed by denoting with $j \in \left\{ 1,..., m+1   \right\}$, where by convention $n_0^0 = 0$ and $n_{m+1}^0 = n$. In the multiple break-point setting, such that $\big( n_1^0,..., n_m^0  \big)$, there are $( m + 1)$ regimes, each corresponding 
to a distinct parameter value $\delta_j^0$, which needs to be estimated. Therefore, when the full sample has $n$ available observations, the aim is to simultaneously estimate the unknown regression coefficients together with the break points. Moreover, the asymptotic theory analysis for the multiple break-points case follow directly from the single break case. However, especially for nonstationary time series regressions the existence of multiple break points might be problematic due to possible changes in the peristence properties of regressors from one regime to the other. Neverthless, the procedure and estimation criterion proposed by \cite{casini2022generalized} can be employed to identify these multiple structural breaks in predictive regression models as well (see, \cite{barigozzi2018simultaneous}). 

\newpage

Therefore, the class of estimators for inference in multiple change-points regressions relies on a certain criterion function (see, \cite{casini2022generalized})
\begin{align}
Q_n \big( \delta ( \uptau_b ), \uptau_b \big) = \sum_{i=1}^{ m+ 1} \sum_{  t = n_i + 1 }^{ n_i }  \big( y_t - f(x_{t-1}) \big)    
\end{align}
As a result, in order to establish the large-sample properties of break-point estimators and test statistics, in a similar spirit as in \cite{casini2022generalized}, we can consider the shrinkage theoretical framework of \cite{bai1998estimating} and \cite{qu2007estimating} (see, also \cite{lavielle2000least} and \cite{nkurunziza2020inference}). On the other hand, necessary modifications of the asymptotic theory analysis is required in order to incorporate the features of the local-to-unity parametrization as well as the use of the proposed sup-Wald type statistics based on the OLS and the IVX estimation.

\section{Structural Change in First Order Autoregressive Models}
\label{Section2}

Break point estimation is an important component in change-point detection problems. In order to obtain some useful insights we begin by considering the statistical properties that correspond to structural break tests and break-point estimators for an AR(1) autoregressive model, where the local-to-unity parametrization is omitted rendering this way a stationary AR$(1)$ model under suitable parameter space restrictions. We follow the study of \cite{pang2021estimating}.

\begin{assumption}
We assume that it holds that
\begin{align}
\mathbb{E} \left(  u_t | \mathcal{F}_{t-1} \right) = 0 \ \ \  \mathbb{E} \left(  u^2_t | \mathcal{F}_{t-1} \right) = 1, \ \textit{almost surely}.  
\end{align}
\end{assumption}

\begin{remark}
Notice that the martingale difference assumption does not allow to capture conditional heteroscedasticity while it is assumed to be weaker than assuming the error sequence $\varepsilon_t$ is independent. Moreover, we are interested to derive the convergence rate of estimates under different econometric conditions. In a linear regression model a relevant framework is given by \cite{shimizu2023asymptotic} (see, also the study of \cite{self1987asymptotic} who investigate the asymptotic properties of MLE estimators and likelihood ratio tests under nonstandard conditions.      
\end{remark}

\begin{corollary}
\label{FCLT}
Given the LUR parametrization of the autoregressive equation in a predictive regression model, then the following joint convergence result holds:
\begin{align}
\left( \frac{1}{n} \sum_{t=1}^{ \floor{ nr } } x_{ t - 1 } u_{t} ,    \frac{1}{ n^2 } \sum_{t=1}^{ \floor{ nr } } x_{ t - 1 }^2 \right) \Rightarrow  \left( \int_0^r J_c (s) dB_u(s) , \int_0^r J_c (s)^2 ds     \right).  
\end{align}
\end{corollary}

\begin{remark}
Notice that Corollary  \ref{FCLT} gives the joint weakly convergence result for the partial sum processes based on a local-to-unity parametrization. Under the maintained hypothesis that the shift exists, that is, $\delta \neq 0$, we will need to derive corresponding moment functionals. 
\end{remark}

\newpage 

\paragraph{Motivation} \cite{katsouris2023testing, katsouris2023predictability} established the asymptotic behaviour of structural break tests in predictive regression models for the sup-Wald OLS and sup-Wald IVX statistic. 
\begin{enumerate}

\item[\textit{(i).}] the limits for both the sup Wald OLS and sup Wald IVX statistics under MI predictors for a predictive regression model with an intercept that shifts converges to the standard NBB.

\item[\textit{(ii).}] the limits for both statistics under LUR predictors for a predictive regression model with an intercept that shifts (after demeaning) converges to non-standard distributions given by Theorem 1 and Theorem 2  in \cite{katsouris2023predictability} respectively. Corresponding results: (a) when the model includes no intercept (i.e., testing only the stability of the slopes) and (b) when the model includes a stable intercept, are established.

\end{enumerate}
We also proved that under the assumption of a known break-point the limiting distributions converge to a nuisance parameter free distribution $\chi^2$ regardless of the persistence properties. However, the studies of  \cite{katsouris2023testing, katsouris2023predictability} didn't consider the asymptotic theory of break-point date estimation, which we aim to present in this article.  

\subsection{Consistency and Limiting Distributions of model coefficients}

\begin{example}
Consider the following AR(1) model without a model intercept,  with a possible shift in the AR parameter at an unknown break-point location in the full sample $k_0 = \floor{ \uptau_0 n   }$,
\begin{align}
y_t = \beta_1 y_{t-1} \boldsymbol{1} \left\{ t \leq k_0 \right\} + \beta_2 y_{t-1} \boldsymbol{1} \left\{ t > k_0 \right\} + \epsilon_t, \ \ t = 1,2,...,n, 
\end{align}
where $\boldsymbol{1} \left\{ . \right\}$ denotes the indicator function, and $\left\{ \epsilon_t, t \geq 1 \right\}$ is a sequence of \textit{i.i.d} random variables. For a given $\tau$, the ordinary least squares estimators of parameters $\beta_1$ and $\beta_2$ are
\begin{align}
\hat{\beta}_1 ( \uptau ) = \frac{ \displaystyle \sum_{t=1}^{[n\uptau] } y_t y_{t-1} }{ \displaystyle  \sum_{t=1}^{[ n\uptau ] } y^2_{t-1} } \ \ \text{and} \ \ \hat{\beta}_2 ( \uptau ) = \frac{ \displaystyle \sum_{ t = [n \uptau ]  + 1}^{ n } y_t y_{t-1} }{ \displaystyle  \sum_{t=[ n \uptau ]  + 1}^{ n } y^2_{t-1} } 
\end{align}
respectively. Then the change-point estimator satisfies $\hat{\uptau}_{n} = \underset{ \uptau \in (0,1) }{ \text{arg min} } \ \text{RSS}_{n}(\uptau)$, where
\begin{align}
\text{RSS}_{n}( \uptau ) &= \sum_{t=1}^{[n \uptau ] } \bigg( y_t - \hat{\beta}_1 ( \uptau ) y_{t-1} \bigg)^2 +  \sum_{t=[n \uptau] + 1}^{ n } \bigg( y_t - \hat{\beta}_2 ( \uptau ) y_{t-1} \bigg)^2.
\end{align}
This expression gives the estimation procedure of the model parameters under the presence of a possible structural break in the full-sample (see, \cite{chong2001structural} and \cite{pang2021estimating}). The estimation method corresponds to the least squares method proposed by \cite{bai1994least} and \cite{bai1993estimation}. 
\end{example}

\newpage 

Therefore, we are interested to estimate the structural parameters $\beta_1$ and $\beta_2$ and the time (or location) of change $\uptau_0$. Furthermore, we estimate the following model:
\begin{align}
\hat{y}_t = \hat{\beta}_1 y_{t-1} \mathbf{1} \left\{ t \leq [ n \hat{\uptau} ] \right\} + \hat{\beta}_2 y_{t-1} \mathbf{1} \left\{ t > [ n \hat{\uptau} ]  \right\} + \epsilon_t, \ \ \ t \in \left\{ 1,...,n \right\}. 
\end{align}
We could also denote with $\mathcal{S}_n = \mathcal{S}_{1n} ( \beta_1, \uptau ) + \mathcal{S}_{2n} ( \beta_2, \uptau )$ where $\mathcal{S}_n \equiv \text{RSS}_{n}( \uptau )$ such that
\begin{align}
\mathcal{S}_{1n} ( \beta_1, \uptau ) &= \sum_{t=1}^{[n \uptau ] } \bigg( y_t - \hat{\beta}_1 ( \uptau ) y_{t-1} \bigg)^2
\\
\mathcal{S}_{2n} ( \beta_1, \uptau ) &=
\sum_{t=[n \uptau] + 1}^{ n } \bigg( y_t - \hat{\beta}_2 ( \uptau ) y_{t-1} \bigg)^2.
\end{align}
Therefore, in practice the criterion function which will need to be employed to estimate the break-point date estimators requires an iterative procedure. In other words, for each $\tau \in \Pi$, we obtain the regression parameter estimators pre$-\floor{ n \tau }$ and post$-\floor{ n \tau }$  such that
\begin{align*}
\hat{\beta}_{jn} (\tau) = \underset{ \beta \subset \Theta }{ \mathsf{arg \ min}  } \ \mathcal{S} \left( \beta, \tau \right)
\end{align*}
for $j \in \left\{ 1,2 \right\}$ respectively. Furthermore, in practice the shift point will be estimated as the sample partition that minimizes the objective function concentrated in $\tau$ such that
\begin{align}
\hat{\uptau}_n = \underset{ \uptau \in \Pi }{ \mathsf{arg \ min} }   \big\{ \mathcal{S}_{1n} \left( \hat{\beta}_{1n} (\uptau ), \uptau \right) +  \left( \hat{\beta}_{2n} (\uptau ), \uptau \right)  \big\}.
\end{align}  
Moreover, we have that $\hat{\vartheta}_n := \big( \hat{\beta}_{1n}, \hat{\beta}_{2n}, \hat{\uptau}_n \big)^{\prime}$, where $\hat{\beta}_{jn} = \hat{\beta}_{jn} \left( \hat{\uptau}_n \right)$ are the coefficient estimators for $j \in \left\{ 1, 2 \right\}$. Moreover, the size of the jump will be estimated and denoted with $\hat{\delta}_n = \left( \hat{\beta}_{1n} - \hat{\beta}_{2n} \right)$. Notice that $\hat{\beta}_1$ has a scaled Dickey-Fuller distribution. The term $\hat{\beta}_2$ can be described as being asymptotically normally distributed with a random variance, which occurs due to the asymptotic limit given by the following expression:
\begin{align}
\sqrt{n} \left( \hat{\beta}_2 - \beta_2 \right) 
= 
\frac{ \displaystyle \frac{ 1 }{ \sqrt{n} } \sum_{ t = \floor{ \uptau_0 n } + 1 }^n y_{t-1} \epsilon_t }{ \displaystyle \frac{ 1 }{ n } \sum_{ t = \floor{ \uptau_0 n } + 1 }^n  y^2_{t-1} } + o_p(1),
\end{align}
The numerator follows a CLT and the denominator converges to a random variable. Moreover, the maintained hypothesis is that the shift exists, which implies that $\delta \neq 0$. 
Consider for example the set of these indicator functions such that $\boldsymbol{1}  \left\{ t \leq \uptau n \right\}$ and $\boldsymbol{1}  \left\{ t \leq \hat{ \uptau } n \right\}$ where $\hat{ \uptau  }$ is an estimator of the unknown break fraction $\uptau$. Then, showing that for the OLS estimator of the predictive regression model it holds that $n \left( \hat{\uptau} - \uptau \right) = \mathcal{O}_p(1)$ implies also that  $
\floor{ \hat{\uptau} n } = \floor{ ( \hat{\uptau} - \uptau )n + \uptau n } = \floor{ O_p(1) + \uptau n  } = \floor{ \uptau n } + o_p(1)$.

\newpage

\begin{theorem}[\cite{chong2001structural}, \cite{pang2021estimating}]
In model (1), if $| \beta_1 | < 1$, $\beta_2 = \beta_{2n} = 1 - c / n$, where $c$ is a fixed constant, and assumptions C1-C3 are satisfied, then the estimators $\hat{\uptau}_n$, $\hat{ \beta}_1 (  \hat{\uptau}_n )$ and  $\hat{ \beta}_2 (  \hat{\uptau}_n)$ are all consistent, and 
\begin{align}
\begin{cases}
\left| \hat{\uptau}_n - \uptau_0 \right| = \mathcal{O}_p ( 1 / n ),   & 
\\
\nonumber
\\
\sqrt{n} \left( \hat{ \beta}_1 (  \hat{\uptau}_n ) - \beta_1  \right) \Rightarrow \mathcal{N} \left( 0, (1 - \beta_1^2) / \uptau  \right) ,  
\\
\nonumber
\\
n \left( \hat{ \beta}_2 (  \hat{\uptau}_n ) - \beta_2  \right)  \Rightarrow \frac{ \displaystyle \frac{1}{2} F^2(W, c, \uptau_0, 1) + c \int_{ \uptau_0 }^1 e^{ 2c(1-t) } F^2(W, c, \uptau_0, t) dt - \frac{1}{2} ( 1- \uptau) }{ \displaystyle \int_{ \uptau_0 }^1 e^{ 2c(1-t)F^2(W, c, \uptau_0, t) dt } }
\end{cases}
\end{align}
\end{theorem}

\medskip

\begin{remark}
When $\underline{\uptau}_0 \leq \tau \leq \bar{\uptau}_0$, $\hat{ \beta }_1 ( \uptau)$, then $\hat{ \beta }_1 ( \uptau )$ converges uniformly to $\beta_1$. The particular result is not surprising because the estimator $\hat{ \beta }_1 ( \uptau )$ is obtained based on the data generating process that corresponds to $y_t = \beta_1 y_{t-1} + u_t$. Moreover, the estimator $\hat{ \beta }_2 ( \uptau )$ converges uniformly to a weighted average of $\beta_1$ and $\beta_2$. The weight depends on the true change point, the true preshift and postshift parameters as well as the location of $\uptau$. Furthermore, if both $\beta_1$ and $\beta_2$ are within the unit boundary then the process is said to be stationary. Thus, it is not difficult to show that all the OLS estimators are consistent in this case. For instance, \cite{bai1994least, bai1993estimation} shows that in the conventional stationary case, the change-point estimator is $n-$consistent. This convergence rate is fast enough to make the limiting distributions of $\hat{\beta}_1$ and $\hat{\beta}_2$ behave as if the true change point $\uptau_0$ is known. Theorem \ref{theorem2} below establishes the asymptotic normality of $\hat{\beta}_1$ and $\hat{\beta}_2$.     
\end{remark}

Relevant frameworks which consider a structural change type estimation and inference of a first-order autoregressive model includes the framework proposed by \cite{kurozumi2023fluctuation}. Although the particular framework corresponds to a fluctuation type monitoring test\footnote{Notice that the monitoring testing approach (see, \cite{chu1996monitoring}, \cite{leisch2000monitoring}, \cite{aue2009delay} and  \cite{horvath2020sequential} among others) corresponds to a different implementation and estimation procedure of structural change in time series regressions. One of the main difference is the use of a historical and a monitoring period during which model estimates and residuals are constructed with the purpose of detecting structural breaks. We leave these considerations as future research. Another aspect worth emphasizing is that in this article the sup-Wald type statistics correspond ton an iterative estimation step in order to construct a sequence of test statistics based on fitting the regression model within the full-sample and comparing the model estimates across two subsamples that correspond to the pre-break and post-break part. } for detecting the presence of explosive behaviour in time series data (see, \cite{arvanitis2018mildly} and \cite{skrobotov2023testing}). On the other hand, in this article we consider the statistical properties of break-point estimators in predictive regression models within the full-sample, so these two aspects are considered to be the main contributions of our study. Furthermore, our work can be useful in relevant applications from the financial economics literature since knowing the exact distributional properties of the break-point estimators when the econometrician employs a predictive regression model either for detecting slope instabilities or testing for predictability robust against parameter instability.

\newpage

\begin{theorem}[\cite{chong2001structural}, \cite{pang2021estimating}]\label{theorem2}
Under regularity conditions, if $| \beta_1 | < 1$ and $| \beta_2 | < 1$, the OLS estimators $\hat{ \uptau}_n$, $\hat{\beta}_1 ( \hat{\uptau}_n )$ and $\hat{\beta}_2 ( \hat{\uptau}_n )$ are consistent estimators and it holds that 
\begin{align}
\left| \hat{\uptau}_n - \uptau_0 \right| &= \mathcal{O}_p \left( \frac{1}{n}  \right), 
\\
\sqrt{n} \left(  \hat{\beta}_1 ( \hat{\uptau}_n ) - \beta_1 \right)  &\to^d \mathcal{N} \left( 0, \frac{ 1 - \beta_1^2 }{ \uptau_0 }  \right)
\\
\sqrt{n} \left(  \hat{\beta}_2 ( \hat{\uptau}_n ) - \beta_2 \right)  &\to^d \mathcal{N} \left( 0, \frac{ 1 - \beta_2^2 }{ 1 - \uptau_0 }  \right)
\end{align}
\end{theorem}
Therefore, we can see from the limit result given on Theorem 1 that $\hat{\beta}_1 ( \hat{\uptau}_n )$ and $\hat{\beta}_2 ( \hat{\uptau}_n )$ are both asymptotically  normally distributed with variance depending on $\beta_1$, $\beta_2$ and $\uptau_0$.

\subsection{The Asymptotic Criterion Function}

Notice that in the model the change point $\uptau_0$ is unknown and has to be estimated by $\hat{\uptau}_n$. Thus, to show the consistency of $\hat{\uptau}_n$, the common practice in the structural break literature is to show that $( 1 / n ) RSS_n( \uptau )$ converges uniformly to a nonstochastic function that has a unique minimum at $\uptau = \uptau_0$. Therefore, we focus on the asymptotic behaviour of the quantity $( 1 / n ) RSS_n( \uptau )$. The following Lemma is useful in deriving the limiting behaviour of the criterion function $( 1 / n ) RSS_n( \uptau )$ and in proving Theorem 1, which follows. 

\begin{lemma}[\cite{pang2021estimating}]
Let $\left\{ y_t \right\}_{t=1}^n$ be generated according to model (1), with $| \beta_1 | < 1$ and $| \beta_2 | < 1$. We have the following asymptotic results
\begin{align}
&\underset{ 0 \leq \uptau_1 \leq \uptau_2 \leq 1 }{ sup } \ \frac{1}{n} \left| \sum_{ [ n \uptau_1 ] }^{ [ n \uptau_2 ]} y_{t-1} \epsilon_t \right| = o_p(1), 
\\
&\frac{1}{n} \sum_{t=1}^{ [ n \uptau_0 ]} y_{t-1}^2 \overset{ p }{ \to } \frac{ \uptau_0 \sigma^2 }{ 1 - \beta_1^2 }, 
\\
&\frac{1}{n} \sum_{ t = [ n \uptau_0 ] + 1 }^{ n } y_{t-1}^2 \overset{ p }{ \to } \frac{ ( 1 - \uptau_0 ) \sigma^2 }{ 1 - \beta_1^2 }, 
\\
&\underset{ 0 \leq \uptau_1 \leq \uptau \leq \uptau_0 }{ sup } \ \left| \sum_{ [ n \uptau ] }^{ [ n \uptau_0 ]} y^2_{t-1} - \frac{ ( \uptau_0 - \uptau) \sigma^2 }{ 1 - \beta_1^2 } \right| = o_p(1), 
\\
&\underset{ 0 \leq \uptau_1 \leq \uptau \leq \uptau_0 }{ sup } \left| \hat{\beta}_1( \uptau ) - \beta_1 \right| = o_p(1),
\\
&\underset{ 0 \leq \uptau_1 \leq \uptau \leq \uptau_0 }{ sup } \left| \hat{\beta}_2( \uptau ) -  \frac{ ( \uptau_0 - \uptau )( 1 - \beta_2^2 ) \beta_1 + ( 1 - \uptau_0 )( 1 - \beta_1^2 ) \beta_2 }{ ( \uptau_0 - \uptau )( 1 - \beta_2^2 ) + ( 1 - \uptau_0 )( 1 - \beta_1^2 ) } \right| = o_p(1),
\end{align}
\end{lemma} 

\newpage 

Consider the criterion function $\left(  1 / n \right) RSS ( \uptau)$ when $\beta_2 = 1$ behaves very differently. Then, under Assumptions (A1)-(A3) we have that 
\begin{align}
\ \ & \left| \frac{1}{n} \sum_{t=1}^{ \floor{ n \uptau_0 } } y_{t-1} \epsilon_t    \right| = o_p(1), 
\\
 \ \ & \frac{1}{n} \sum_{ t = 1\floor{ n \uptau_0 } + 1 }^{ n }  y_{t-1} \epsilon_t  \Rightarrow \frac{( 1 - \uptau_0 ) \sigma^2 }{2} \left( B^2(1) - 1 \right) = \mathcal{O}_p(1), 
\end{align}
Following the econometric framework of \cite{chong2001structural} and \cite{pang2021estimating}, further limit results which will need to be established in the case one replaces the assumption of a stationary autoregressive coefficient with the local-to-unity parametrization includes: 
\begin{align*}
&\underset{ \uptau_1 \leq \uptau \leq \uptau_0 }{ \text{sup} } \left| \hat{\beta}_1 ( \uptau ) - \beta_1  \right| = o_p(1),
\ \ \ \underset{ \uptau_1 \leq \uptau \leq \uptau_0 }{ \text{sup} } \left| \hat{\beta}_2 ( \uptau ) - 1  \right| = \mathcal{O}_p \left( \frac{1}{n} \right),
\ \ \ \underset{ \uptau_1 \leq \uptau \leq \uptau_0 }{ \text{sup} } \left| \hat{\beta}_2 ( \uptau ) - \beta_1  \right| = \mathcal{O}_p( 1 ),
\\
&\underset{ \uptau_1 \leq \uptau \leq \uptau_0 }{ \text{sup} } \frac{1}{n}  \left| \sum_{ \floor{ n \uptau_0 } + 1 }^{ \floor{ n \uptau } } y_{t-1} \epsilon_t \right| = \mathcal{O}_p( 1 ),
\ \ \ \ \underset{ \uptau_1 \leq \uptau \leq \uptau_0 }{ \text{sup} } \frac{1}{n}  \left| \sum_{ \floor{ n \uptau } + 1 }^{ n } y_{t-1} \epsilon_t \right| = \mathcal{O}_p( 1 ),
\\
&\left| \hat{\beta}_1 ( \uptau ) - 1 \right| = \mathcal{O}_p \left( \frac{1}{ \sqrt{n} } \right),
\ \ \ \  \left| \hat{\beta}_2 ( \uptau ) - 1 \right| = \mathcal{O}_p \left( \frac{1}{ n } \right),
\end{align*}
Notice that there is an asymptotic gap between $( 1 / n) RSS_n( \uptau_0 )$ and  $( 1 / n ) RSS ( \uptau )$. Thus to examine the consistency of $\hat{\beta}_1$, we have to investigate the transitional behaviour of $( 1 / n ) RSS_n ( \uptau )$. Note that for any constant $c > 0$, 
\begin{align}
\hat{\beta}_1 \left( \uptau_0 + c n^{ \alpha - 1 } \right) 
= 
\theta_n \left( \alpha, c \right) \left( \beta_1 + \frac{ \displaystyle \sum_{t = 1}^{k_0}  y_{t-1} \epsilon_t }{ \displaystyle \sum_{t = 1}^{k_0} y_{t-1}^2 } \right) 
+ 
\big( 1 - \theta_n \left( \alpha, c \right) \big) \left( 1 + \frac{ \displaystyle \sum_{t = k_0 + 1 }^{ k_0 + \floor{ c n^{\alpha} } }  y_{t-1} \epsilon_t }{ \displaystyle \sum_{t = k_0 + 1 }^{ k_0 + \floor{ c n^{\alpha} } }  y_{t-1}^2 } \right), 
\end{align}
where 
\begin{align}
\theta_n \left( \alpha, c \right) 
= 
\left( \frac{ \displaystyle \sum_{t = 1}^{k_0}  y^2_{t-1} \epsilon_t }{ \displaystyle\sum_{t = 1}^{k_0}  y^2_{t-1} + \sum_{t = k_0 + 1 }^{ k_0 + \floor{ c n^{\alpha} } }  y_{t-1}^2 } \right) 
\end{align}
When $\alpha < 1 / 2$, it holds that $\theta_n \left( \alpha, c \right) \overset{ p }{ \to } 1, \ \ \ \hat{\beta}_1 \left( \uptau_0 + c n^{ \alpha - 1 } \right) \overset{ p }{ \to } \beta_1$, and $\frac{1}{n} RSS_n \left( \uptau_0 + c n^{ \alpha - 1 } \right) \overset{ p }{ \to } \sigma^2$. Moreover, for $\frac{1}{2} < \alpha < 1$ we have that 
\begin{align}
\theta_n \left( \alpha, c \right) &\overset{ p }{ \to } 0, \ \ \ \hat{\beta}_1 \left( \uptau_0 + c n^{ \alpha - 1 } \right) \overset{ p }{ \to } 1,
\\
\frac{1}{n} RSS_n \left( \uptau_0 + c n^{ \alpha - 1 } \right) &\overset{ p }{ \to } \sigma^2 + \frac{ ( 1 - \beta_1 ) \uptau_0 \sigma^2 }{ 1 + \beta_1 }. 
\end{align}

\newpage 

The above results simply imply that if the convergence rate of $\hat{\uptau}_n$ is faster that $n^{1/2}$ then $\hat{\beta}_1$ will be consistent, otherwise it will be inconsistent. Next, we derive the asymptotic behaviour of the criterion function $( 1 / n ) RSS_n ( \uptau )$.  Another useful quantity is as below: 
\begin{align}
\hat{\beta}_1 \left( \uptau_0 + \frac{c}{\sqrt{n}} \right) 
&= 
\theta_n \left( \frac{1}{2}, c \right) \left( \beta_1 + \frac{ \displaystyle \sum_{t = 1}^{k_0}  y_{t-1} \epsilon_t }{ \displaystyle \sum_{t = 1}^{k_0} y_{t-1}^2 } \right) 
\nonumber
+ 
\left[ 1 - \theta_n \left( \frac{1}{2}, c \right) \right] \left( 1 + \frac{ \displaystyle \sum_{t = k_0 + 1 }^{ \floor{ k_0 + c \sqrt{n} } }  y_{t-1} \epsilon_t }{ \displaystyle \sum_{t = k_0 + 1 }^{ k_0 +  c \sqrt{n} }  y_{t-1}^2 } \right), 
\\
\theta_n \left( \frac{1}{2},  c \right) 
&= 
\left( \frac{ \displaystyle \sum_{t = 1}^{k_0}  y^2_{t-1} \epsilon_t }{ \displaystyle\sum_{t = 1}^{k_0}  y^2_{t-1} + \sum_{t = k_0 + 1 }^{ \floor{ k_0 + c \sqrt{n} } } y_{t-1}^2 } \right) 
\end{align} 
Clearly, the challenging task here is that when we impose that assumption that the autoregressive coefficient of the first order autoregression model is expressed via the local-to-unity parametrization, then we expect that the break-point estimators will depend on the nuisance parameter of persistence. In other words, the fact that the asymptotic distribution of these break-point estimators will not be nuisance parameter-free can be challenging when critical values are needed (e.g., case of the monitoring scheme).  We leave these considerations for future research. 

Following the existing literature  to find the limiting distribution of $\hat{\beta}_1 ( \hat{\uptau}_n )$, in the case of the stationary AR$(1)$ model, notice that $\left( \hat{\uptau}_n - \uptau_0 \right) = \mathcal{O}_p \left( n^{-1} \right)$ and it holds that
\begin{align}
\sqrt{n} \left( \hat{\beta}_1 ( \hat{\uptau}_n ) - \hat{\beta}_1 ( \uptau_0 ) \right)
=
\sqrt{n} \left( \frac{ \displaystyle \sum_{t = 1 }^{ \floor{ n\hat{\uptau}_n } }  y_t y_{t-1} }{ \displaystyle \sum_{t = 1 }^{ \floor{ n \hat{\uptau}_n } }  y_{t-1}^2 } - \frac{ \displaystyle \sum_{t = 1 }^{ \floor{ n \uptau_0 } }  y_t y_{t-1} }{ \displaystyle \sum_{t = 1 }^{ \floor{ n \uptau_0 } }  y_{t-1}^2 }   \right)
\end{align}
From the above derivations we can see that $\hat{\beta}_1 ( \hat{\uptau}_n )$ and $\hat{\beta}_1 ( \uptau_0 )$ have the same asymptotic distribution, since $\left\{ y_{t-1} \epsilon_t, \mathcal{F}_t \right\}_{t=1}^{\floor{ n \uptau_0 } }$ is a martingale difference sequence, with $\mathbb{E} \left[ y_{t-1} \epsilon_t | \mathcal{F}_{t-1} \right] = 0$ and $\sum _{t=1}^{ \floor{ n \uptau_0 } } \mathbb{E} \left[  \left( y_{t-1} \epsilon_t \right)^2| \mathcal{F}_{t-1}  \right] \overset{ p }{ \to } \frac{ \sigma^4 }{ 1 - \beta_1^2 } < \infty$. Applying the central limit theorem for martingale difference sequences and the fact that 
\begin{align}
\frac{1}{n} \sum _{t=1}^{ \floor{ n \uptau_0 } }  y_{t-1}^2 &\overset{ p }{ \to } \frac{ \uptau_0 \sigma^2  }{ 1 - \beta_1^2 }, 
\\
\sqrt{n} \left( \hat{\beta}_1 ( \hat{\uptau}_n ) - \beta_1 \right) &\overset{ d }{ = } \sqrt{n} \left( \hat{\beta}_1 ( \hat{\uptau}_0 ) - \beta_1 \right) = \frac{ \displaystyle \frac{1}{ \sqrt{n} } \sum_{t = 1 }^{ \floor{ n \uptau_0 } }  y_{t-1} \epsilon_t }{ \displaystyle \frac{1}{n}   \sum_{t = 1 }^{ \floor{ n \uptau_0 } }  y_{t-1}^2 }  \overset{ d }{ \to } \mathcal{N} \left( 0, \frac{1 - \beta_1^2 }{ \uptau_0 }  \right). 
\end{align}

\newpage 

Next, to find the limiting distribution of $\hat{\beta}_2 ( \hat{\uptau}_n )$, notice that $\left( \hat{\uptau}_n - \uptau_0 \right) = \mathcal{O}_p ( \frac{1}{n} )$ and
\begin{align*}
&n \left( \hat{\beta}_2 ( \hat{\uptau}_n ) - \hat{\beta}_2 ( \uptau_0 ) \right)
=
n \left( \frac{ \displaystyle \sum_{t = \floor{ n\hat{\uptau}_n }  + 1 }^{ n }  y_t y_{t-1} }{ \displaystyle \sum_{t = \floor{ T\hat{\uptau}_n }  + 1 }^{ n }  y_{t-1}^2 } - \frac{ \displaystyle \sum_{t = \floor{ n\uptau_0 }  + 1 }^{ n }  y_t y_{t-1} }{ \displaystyle \sum_{t = \floor{ n \uptau_0 }  + 1 }^{ n }   y_{t-1}^2 }   \right)
\\
&= 
\boldsymbol{1} \left\{ \hat{\uptau}_n \leq \uptau_0 \right\} n \left( - \frac{ \displaystyle \sum_{ t = \floor{ n\hat{\uptau}_n }  + 1 }^{ \floor{ n\uptau_0 } } y^2_{t-1} }{ \displaystyle \sum_{ t = \floor{ n\hat{\uptau}_n }  + 1  }^{ n }  y_{t-1}^2 } \frac{ \displaystyle \sum_{t = \floor{ n\uptau_0 } + 1 }^{ n } y_{t-1} \epsilon_t }{ \displaystyle \sum_{ t = \floor{ n\uptau_0 }  + 1 }^{ n } y_{t-1}^2 } + \frac{ \displaystyle \sum_{ t = \floor{ n\hat{\uptau}_n }  + 1 }^{ \floor{ n\uptau_0 } } y_{t-1} \epsilon_t  }{ \displaystyle \sum_{ t = \floor{ n\hat{\uptau}_n }  + 1 }^{ \floor{ n\uptau_0 } } y_{t-1}^2 } \right) 
+ 
\boldsymbol{1} \left\{ \hat{\uptau}_n \leq \uptau_0 \right\} n \left( \beta_1 - \beta_2 \right) \frac{ \displaystyle \sum_{ t = \floor{ n\hat{\uptau}_n }  + 1 }^{ \floor{ n\uptau_0 } }  y_{t-1}^2  }{ \displaystyle \sum_{ t = \floor{ n\hat{\uptau}_n }  + 1 }^{ \floor{ n\uptau_0 } } y_{t-1}^2 }         
\end{align*}
It can be prove that both $\hat{\beta}_1 ( \hat{\uptau}_n )$ and $\hat{\beta}_2 ( \uptau_0 )$ have the same asymptotic distribution, which implies  
\begin{align}
n \left( \hat{\beta}_2 ( \hat{\uptau} ) - 1 \right) \overset{ d }{ = } n \left( \hat{\beta}_2 ( \uptau_0 ) - \beta_2 \right) 
= 
\frac{ \displaystyle \frac{1}{n} \sum_{ t = \floor{ n \uptau_0 }  + 1 }^{ n }  y_{t-1} \epsilon_t  }{ \displaystyle  \frac{1}{n^2} \sum_{ t = \floor{ n \uptau_0 }  + 1 }^{ n } y_{t-1}^2 }
\Rightarrow
\frac{ \displaystyle B^2(1) - 1 }{ \displaystyle 2 ( 1 - \uptau_0 ) \int_0^1 B^2 }.
\end{align}
By the central limit theorem for martingale difference sequences and by independence of the two martingale differences given previously, we have that 
\begin{align*}
&\frac{1}{ \sqrt{n} } \sum_{ t = \floor{ n \uptau_0 }  + 1 }^{ n } y_{t-1} \epsilon_t \overset{ d }{ \to } \mathcal{N} \left( 0, \frac{ \sigma^4 }{ 1 - \beta_2^2 } \right) 
\\
&\frac{1}{ n } \sum_{ t = \floor{ n \uptau_0 }  + 1 }^{ \floor{ n \uptau } } y^2_{t-1} 
= 
\frac{1}{n} \sum_{ t = \floor{ n \uptau_0 }  + 1 }^{ \floor{ n \uptau } } \left( \beta_2^{t - k_0 - 1} y_{k_0} +  \sum_{ t = \floor{ n \uptau_0 }  + 1 }^{ t - 1 } \beta_2^{t - i - 1} \epsilon_i \right)^2 
\\
&=
\left( \frac{ y_{ k_0 } }{ \sqrt{n}} \right)^2 \sum_{ t = \floor{ n \uptau_0 }  + 1 }^{ \floor{ n \uptau } }  \beta_2^{ 2 (t - k_0 - 1) } 
+
2 \frac{ y_{ k_0 } }{ T } \sum_{ t = \floor{ n \uptau_0 }  + 1 }^{ \floor{ n \uptau } } \left( \beta_2^{ t - k_0 - 1 } \sum_{ t = \floor{ n \uptau_0 }  + 1 }^{ \floor{ n \uptau } } \beta_2^{ t - i - 1 } \epsilon_i \right)
+ 
\frac{1}{n} \sum_{ t = \floor{ n \uptau_0 }  + 1 }^{ \floor{ n \uptau } } \left( \sum_{ t = \floor{ n \uptau_0 }  + 1 }^{ t - 1 } \beta_2^{ t - i - 1 } \epsilon_i \right)^2.    
\end{align*}
The first term above weakly converges to $\left[ \sigma^2 B^2 ( \uptau_0 ) \right] / \left( 1 - \beta_2^2 \right)$. Furthermore, because $\left| y_{k_0} / \sqrt{n} \right| = \mathcal{O}_p(1)$ and $\left( 1 / \sqrt{n} \right) \text{sup}_{ t > k_0 } \left| \epsilon_t \right| = o_p(1)$, then we can show that the second term is bounded by 
\begin{align*}
&\leq 2 \left|  \frac{y_{k_0} }{ \sqrt{n} } \right| \underset{ \uptau \in \Pi }{ \text{sup} } \left| \sum_{ t = \floor{ n \uptau_0 }  + 1 }^{ \floor{ n \uptau } } \beta_2^{ t - k_0 - 1} \frac{ 1 - \beta_2^{ t - k_0 - 1 } }{ 1 - \beta_2 } \right| \frac{1}{ \sqrt{n} } \underset{ t > k_0 }{ \text{sup} }  \left| \epsilon_t \right| 
\\
&\leq 
2 \left| \frac{y_{k_0} }{ \sqrt{n} } \right| 
\left( \sum_{ t = k_0 + 1 }^{ \infty } \left| \frac{ \beta_2^{ t - k_0 - 1} }{ 1 - \beta_2 } \right|  + \sum_{ t = k_0 + 1 }^{ \infty } \left| \frac{ \beta_2^{ 2 (t - k_0 - 1 ) } }{ 1 - \beta_2 } \right|  \right) \frac{1}{ \sqrt{n} } \underset{ t > k_0 }{ \text{sup} }  \left| \epsilon_t \right|
\\
&\leq 
2 \left| \frac{y_{k_0} }{ \sqrt{n} } \right| 
\left( \frac{1}{ ( 1 - | \beta_2 | )^2 } + \frac{1}{ ( 1 - | \beta_2 | )( 1 - \beta_2^2 ) } \right) \frac{1}{ \sqrt{n} } \underset{ t > k_0 }{ \text{sup} }  \left| \epsilon_t \right| = o_P(1). 
\end{align*}

\newpage 






\newpage

\section{Structural Change in the Predictive Regression Model}
\label{Section3}

Although in this paper we assume that the innovation sequence of the model has a linear process representation, other studies in the literature imposes a NED condition when developing structural break tests in stationary time series regression models (see, \cite{ling2007testing},  \cite{kim2020mean} and \cite{lee2014functional}). 
Notice that the IVX estimator has the property that decorrelates the system and therefore conventional invariance principles and weak convergence results hold without requiring to consider a topological convergence in a different space. Furthermore, another important feature of the predictive regression model is that one can incorporate serial correlation in the error term and therefore using the IVX instrumentation mixed Gaussianity distributional converges still holds.  

\medskip

\begin{example}
Consider again the predictive regression model given below:
\begin{align}
y_t = ( \alpha_1 + \boldsymbol{\beta}_1 x_{t-1} ) \boldsymbol{1} \left\{ t \leq k  \right\} +  ( \alpha_2 + \boldsymbol{\beta}_2 x_{t-1} ) \boldsymbol{1} \left\{ t > k  \right\}  + u_t  
\end{align}
where the autoregressive coefficient has an autoregressive structure based on the local-to-unity parametrization given by the expression below:
\begin{align}
\boldsymbol{x}_t = \left( \boldsymbol{I} - \frac{ \boldsymbol{C}_p  }{n} \right) \boldsymbol{x}_{t-1} + \boldsymbol{v}_t
\end{align}
\end{example}
At this point, we begin our asymptotic theory analysis by investigating the corresponding expressions of the criterion function when the functional form of the linear predictive regression model with a conditional mean function is employed. We write the residual sum of squares as below: 
\begin{align}
\mathcal{S}_n ( \uptau ) 
&= 
\sum_{t=1}^{ \floor{ n \uptau } } \bigg( u_t - \left( \hat{\beta}_1 ( \uptau ) - \beta_1 \right) x_{t-1} \bigg)^2 
+ 
\sum_{t = \floor{ T \uptau } + 1 }^{ \floor{ n \uptau_0 } } \bigg( u_t - \left( \hat{\beta}_2 ( \uptau ) - \beta_1 \right) x_{t-1} \bigg)^2
\nonumber
\\
&+
\sum_{t = \floor{ n \uptau_0 } + 1 }^{ n } \bigg( u_t - \left( \hat{\beta}_2 ( \uptau ) - \beta_2 \right) x_{t-1} \bigg)^2
\end{align}
which after expanding out, the RSS expression can be written as below: 
\begin{align}
\mathcal{S}_n ( \uptau ) 
&= 
\sum_{ t = 1 }^{ n } u_t^2 - 2 \sum_{t=1}^{ \floor{ n \uptau } } \left( \hat{\beta}_1 ( \uptau ) - \beta_1 \right) x_{t-1}  + \left( \hat{\beta}_1 ( \uptau ) - \beta_1 \right)^2 \sum_{t=1}^{ \floor{ n \uptau } } x_{t-1}^2 
\nonumber
\\
\nonumber
\\
&- 
2 \left( \hat{\beta}_2 ( \uptau ) - \beta_1 \right) \sum_{t = \floor{ n \uptau } + 1 }^{ \floor{ n \uptau_0 } } x_{t-1} u_t + \left( \hat{\beta}_2 ( \uptau ) - \beta_1 \right)^2 \sum_{t = \floor{ n \uptau } + 1 }^{ \floor{ n \uptau_0 } } x_{t-1}^2 
\nonumber
\\
\nonumber
\\
&- 
2 \left( \hat{\beta}_2 ( \uptau ) - \beta_2 \right) \sum_{t = \floor{ n \uptau_0 } + 1 }^{ n }  x_{t-1} u_t + \left( \hat{\beta}_2 ( \uptau ) - \beta_2 \right)^2 \sum_{t = \floor{ n \uptau_0 } + 1 }^{ n } x_{t-1}^2.  
\end{align}

\newpage

Furthermore, we need to determine the validity of the following expansion in the case of the predictive regression model. 
\begin{align}
&=
\sum_{t=1}^{ n } \epsilon_t^2 
- \frac{ \displaystyle  \left( \sum_{t = 1 }^{  \floor{ n \uptau } } y_{t-1} \epsilon_t \right)^2 }{ \displaystyle \sum_{t = 1 }^{  \floor{ n \uptau } }  y_{t-1}^2 } 
- 2 \left( \hat{\beta}_2 ( \uptau ) - \beta_1 \right) \sum_{t = \floor{ n \uptau } + 1 }^{ \floor{ n \uptau_0 } } y_{t-1} \epsilon_t
\nonumber
\\
&+
\left( \hat{\beta}_2 ( \uptau ) - \beta_1 \right)^2   \sum_{t = \floor{ n \uptau } + 1 }^{ \floor{ n \uptau_0 } } y_{t-1}^2 - 2 \left( \hat{\beta}_2 ( \uptau ) - \beta_2 \right) \sum_{t = \floor{ n \uptau_0 } + 1 }^{ \floor{ n \uptau } } y_{t-1} \epsilon_t
\nonumber
\\
&+ 
\left( \hat{\beta}_2 ( \uptau ) - \beta_2 \right)^2   \sum_{t = \floor{ n \uptau_0 } + 1 }^{ n } y_{t-1}^2.
\end{align}
Furthermore, we consider the asymptotic behaviour of the break-estimators when we are testing the stability of the model parameters of the  predictive regression model based on the IVX estimator of the predictive regression model. In particular, within our framework we implement the IVX filter proposed by PM (2009) which implies the use of a mildly integrated instrumental variable that has the following form 
\begin{align}
Z_{tn} = \sum_{ j = 0 }^{ t- 1} \left( 1 - \frac{c_z}{ n^{\delta} }   \right) \left( x_{t-j} - x_{t-j-1} \right), 
\end{align}
for some $c_z > 0$ and $0 < \delta < 1$, where $\delta$ is the exponent rate of the persistence  coefficient which corresponds to the instrumental variable. Notice that the above filtering method transforms the autoregressive process $x_t$, which can be either stable or unstable, into a mildly integrated process which is less persistent a Nearly Integrated array such as the case of $x_t$. 

\medskip

\begin{remark}
Under the alternative hypothesis, we assume the possible presence of a single structural break on the parameters of the regressors coefficients of the predictive regression model, which implies two regimes of slope coefficients. An important aspect to emphasize in our framework is that we treat the model intercept differently with respect to the other parameters when testing for structural breaks. The reason is that keeping the model intercept unchanged provides an initial condition to the predictive regression model for both regimes while focusing on testing for a possible parameter instability in the remaining parameters, simplifying also some of the derivations for the development of the asymptotic theory of the test statistics. Nevertheless, once the null hypothesis $\mathbb{H}_0$, is rejected, a practitioner needs to locate the break point or the break point fraction $\hat{\uptau}$ relative to the sample size and this is the importance of the proposed framework\footnote{Under the alternative hypothesis, one can also consider the implementation of asymptotically most powerful tests (see, \cite{choi1996asymptotically}). Specifically, for the case of nonstationary time series regression such as the autoregressive or predictive regressions a different stream of literature consider the semiparametric approach as in the studies of \cite{jansson2008semiparametric}, \cite{werker2022semiparametric} and  \cite{andersen2021consistent}. One can then consider extending those frameworks for the power envelops for break-point estimations when a possible structural break exists.}.    
\end{remark}

\newpage 

\begin{remark}
Suppose that the alternative hypothesis denoted with $H_1$ is true. Furthermore, it is well-known in mathematical statistics that the optimal test (NP test or likelihood-ratio test) is described by Newman-Pearson lemma and the critical region of this test is defined as follows: 
\begin{align}
C_{NP} ( \alpha ) = \left\{ x: \frac{ \mu_U (x) }{ \nu(x) } \leq \lambda_{\alpha}  \right\}, 
\end{align}
where $\alpha \in (0,1)$ is the significance level and the constant $\lambda_{ \alpha }$ is chosen in such a way that $\mu_U \big( C_{NP} ( \alpha ) \big) = \alpha$. Notice also that when we consider the asymptotic behaviour of the tests for large $n$ any effects occurring  due to non-randomized tests will be negligible. 
\end{remark}

\subsection{Estimating Breaks in Predictive Regressions}

\begin{example}
Consider the following predictive regression model 
\begin{align}
\label{model1}
y_t 
= 
\begin{cases}
\beta_1 x_{t-1} + \epsilon_t, & 1 \leq t \leq k_1^0, \\
\beta_2 x_{t-1} + \epsilon_t, & k_1^0 + 1 \leq t \leq k_2^0,  \\
\beta_3 x_{t-1} + \epsilon_t, & k_2^0 + 1 \leq t \leq n
\end{cases}
\ \ \ \ 
x_t 
= 
\begin{cases}
\rho_1 x_{t-1} + u_t, & 1 \leq t \leq k_1^0, \\
\rho_2 x_{t-1} + u_t, & k_1^0 + 1 \leq t \leq k_2^0,  \\
\rho_3 x_{t-1} + u_t, & k_2^0 + 1 \leq t \leq n
\end{cases}
\end{align}
with $\rho_1 = \left(  1 - \frac{c_1}{n} \right)$, $\rho_2 = \left(  1 + \frac{c_2}{n} \right)$ and $\rho_3 = \left(  1 - \frac{c_3}{n} \right)$    where $c_1, c_2, c_3 > 0$. To develop an estimation procedure for the model given by \eqref{model1}-\eqref{model2}, we first compute the difference of the residual sums of squares at the break-points $k_1^0$ and $k_2^0$. Let RSS$(\tau)$ be the residual sum of squares on the date $[ \tau n ]$ then, the following theorem can be shown to hold (see, \cite{pang2021estimating}).
\end{example}
\begin{theorem}[\cite{pang2021estimating}]
For model \eqref{model1}-\eqref{model2}, we have that 
\begin{align}
\text{RSS}(\uptau_1^0) - \text{RSS}(\uptau_2^0)= \eta_1 \left( \beta_2 - \beta_1 \right) + \eta_2 \left( \beta_3 - \beta_2 \right) + \eta_3 \left( \beta_2 - \beta_1 \right)^2 +  \eta_4 \left( \beta_3 - \beta_2 \right)^2 + \Omega_n,
\end{align} 
\begin{align*}
\text{RSS}(\uptau_1^0) - \text{RSS}(\uptau_2^0) &=
\begin{cases}
\eta_1 & =  2 \left( \displaystyle \frac{ \sum_{t=1}^{k_1^0} x_{t-1} \epsilon_t }{ \sum_{t=1}^{k_1^0} x^2_{t-1}  } - \frac{ \sum_{t=1}^{k_2^0} x_{t-1} \epsilon_t }{ \sum_{t=1}^{k_2^0} x^2_{t-1}  } \right) \displaystyle \sum_{t=1}^{k_1^0} x^2_{t-1},
\\
\eta_2 & =  2 \left( \displaystyle \frac{ \sum_{t=k_1^0 + 1}^{k_2^0} x^2_{t-1} \sum_{t=k_2^0 + 1}^{n} x_{t-1} \epsilon_t - \sum_{t=k_1^0 + 1}^{k_2^0}  x_{t-1} \epsilon_t \sum_{t=k_2^0 + 1}^{n} x^2_{t-1} }{ \sum_{t=k_1^0 + 1}^{n} x^2_{t-1}  } \right),
\\
\eta_3 & =  - \displaystyle \frac{ \sum_{t=1}^{k_1^0} x^2_{t-1}   \sum_{t=k_1^0 + 1}^{k_2^0} x^2_{t-1}}{  \sum_{t=1}^{k_2^0} x^2_{t-1}    }
\\
\eta_4 & =  - \displaystyle \frac{ \sum_{t=k_1^0 + 1}^{k_2^0} x^2_{t-1}   \sum_{t=k_2^0 + 1}^{n} x^2_{t-1}}{  \sum_{t=k_1^0 + 1}^{n} x^2_{t-1}    } 
\end{cases}
\\
\Omega_n &= \displaystyle  \frac{ \left( \sum_{t=1}^{k_2^0} x_{t-1} \epsilon_t \right)^2}{ \sum_{t=1}^{k_2^0} x^2_{t-1} }  + \frac{ \left( \sum_{t=k_2^0 + 1}^{n} x_{t-1} \epsilon_t \right)^2}{ \sum_{t=k_2^0 + 1}^{n} x^2_{t-1} } - \frac{ \left( \sum_{t=1}^{k_1^0} x_{t-1} \epsilon_t \right)^2}{ \sum_{t=1}^{k_1^0} x^2_{t-1} } - \frac{ \left( \sum_{t=k_1^0 + 1}^{n}  x_{t-1} \epsilon_t \right)^2}{ \sum_{t=k_1^0 + 1}^{n}  x^2_{t-1} }.
\end{align*}
\end{theorem}

\newpage 

\begin{proof}
\begin{align}
\text{RSS}_{1,n}( \uptau_1^0 ) = \sum_{t=1}^{ k_1^0 } \bigg( y_t - \hat{\beta}_1 ( \uptau_1^0 ) x_{t-1} \bigg)^2 +  \sum_{t = k_1^0 + 1}^{ k_2^0 } \bigg( y_t - \hat{\beta}_2 ( \uptau_1^0 ) x_{t-1} \bigg)^2.
\end{align}
\begin{align*}
\sum_{t=1}^{ k_1^0 } \bigg( y_t - \hat{\beta}_1 ( \uptau_1^0 ) x_{t-1} \bigg)^2 
&=
\sum_{t=1}^{ k_1^0 } y^2_t - 2 \sum_{t=1}^{ k_1^0 } y_t \hat{\beta}_1 ( \uptau_1^0 ) x_{t-1} + \sum_{t=1}^{ k_1^0 } \hat{\beta}_1^2 ( \uptau_1^0 ) x^2_{t-1}
\end{align*}

Therefore, we have that
\begin{align*}
\sum_{t=1}^{ k_1^0 } \bigg( ... \bigg)^2 
&=
\sum_{t=1}^{ k_1^0 } \left\{ (  \beta_1 x_{t-1} + \epsilon_t )^2 - 2 (  \beta_1 x_{t-1} + \epsilon_t ) \frac{ \sum_{t=1}^{k_1^0  } (  \beta_1 x_{t-1} + \epsilon_t ) x_{t-1} }{  \sum_{t=1}^{k_1^0  } x^2_{t-1} } x_{t-1} +  \left[ \frac{ \sum_{t=1}^{k_1^0  } (  \beta_1 x_{t-1} + \epsilon_t ) x_{t-1} }{  \sum_{t=1}^{k_1^0  } x^2_{t-1} }  \right]^2 x_{t-1}^2 \right\}
\\
&= \sum_{t=1}^{ k_1^0 } \left\{  \beta_1^2 x_{t-1}^2  + 2 \beta_1 x_{t-1} \epsilon_t + \epsilon_t^2  \right\} - 2 \frac{ \left[ \sum_{t=1}^{ k_1^0 } (  \beta_1 x_{t-1} + \epsilon_t ) x_{t-1} \right]^2}{  \sum_{t=1}^{k_1^0  } x^2_{t-1}  } + \frac{ \left[ \sum_{t=1}^{ k_1^0 } (  \beta_1 x_{t-1} + \epsilon_t ) x_{t-1} \right]^2}{  \sum_{t=1}^{k_1^0  } x^2_{t-1}  }
\\
&= 
\sum_{t=1}^{ k_1^0 } \left\{  \beta_1^2 x_{t-1}^2  + 2 \beta_1 x_{t-1} \epsilon_t + \epsilon_t^2  \right\} -  \frac{ \left[ \sum_{t=1}^{ k_1^0 } (  \beta_1 x_{t-1} + \epsilon_t ) x_{t-1} \right]^2}{  \sum_{t=1}^{k_1^0  } x^2_{t-1}  }
\\
&= 
\beta_1^2 \sum_{t=1}^{ k_1^0 } x_{t-1}^2 + 2 \beta_1 \sum_{t=1}^{ k_1^0 } x_{t-1} \epsilon_t + \sum_{t=1}^{ k_1^0 } \epsilon_t^2 - \beta_1^2 \sum_{t=1}^{ k_1^0 } x_{t-1}^2 - \beta_1^2 \frac{ \left( \sum_{t=1}^{k_1^0} x_{t-1} \epsilon_t \right)^2 }{ \sum_{t=1}^{k_1^0} x^2_{t-1} }
\\
&= 
2 \beta_1 \sum_{t=1}^{ k_1^0 } x_{t-1} \epsilon_t  + \sum_{t=1}^{ k_1^0 } \epsilon_t^2 - \beta_1^2 \frac{ \left( \sum_{t=1}^{k_1^0} x_{t-1} \epsilon_t \right)^2 }{ \sum_{t=1}^{k_1^0} x^2_{t-1} }
\end{align*}
\end{proof}

\subsection{Break Estimation Procedure}

The case study we illustrate in the previous section is based on the specific approach proposed by \cite{pang2021estimating}. In this article since the estimation methodology of model coefficients is taken into consideration, then we need to establish the asymptotic theory separately for the break-point estimator which based on the OLS optimization versus the break-point IVX optimization, in order to evaluate the consistency and statistical properties of the estimated break-point. Our objective in is to evaluate the properties of $k_1$ which is an estimator of the location of the break-point $k_1^0$ in the slope parameters and intercept of the predictive regression model (see, \cite{kostakis2015robust}).

\newpage

In particular, the estimator of the break-point is obtained by minimizing the concentrated sum of squares errors function as
\begin{align}
S_{1n} (k) = \sum_{t=1}^k \big( y_t - x_t^{\prime} \hat{\beta}_1 (k) \big)^2 + \sum_{t=k+1}^n \big( y_t - x_t^{\prime} \hat{\beta}_2 (k) \big)^2     
\end{align}
where $\hat{\beta}_1 (k)$ and $\hat{\beta}_2 (k)$ denote the least squares estimators of the slope parameters within each regime for given $k$. Alternatively, we can also reformulate $\hat{k}_1$ as below
\begin{align}
\hat{k}_1 = \mathsf{arg \ max}_k \ G_{1n} (k) \ \ \ \text{where} \ \ S_n - S_{1n} (k) \ \ \text{and} \ \ S_n = \sum_{t=1}^n \big( y_t - x_t^{\prime} \hat{\beta} \big)^2      
\end{align}
where $S_n$ denotes the full sample sum of squared errors. Thus, we are interested to establish the weak consistency of $\hat{\uptau} = \hat{k}_1/n$ under a certain set of assumptions. This formulation allows us to establish the weak consistency of the break fraction $\hat{\uptau} = \hat{k}_1 / n$ under certain set of assumptions. Thus, to have more meaningful power comparisons we study in more details the asymptotic behaviour of the break estimators under the alternative hypothesis based on the two estimators. To do this, we follow the methodology proposed with the framework of \cite{pang2021estimating}, but in our setting we focus in the case of a single unknown break-point. We consider separately the break-point estimator based on the OLS versus the IVX estimators to evaluate the consistency of the estimated break-point and obtain corresponding convergence rates\footnote{The finite sample properties of $\hat{\uptau}_1$ is important not only because of the direct economic implications that the accurate dating of a structural break in the mean may entail but also for the subsequent analysis which could involve the search for further breaks in the variance, typically based on the residual sequence (see, \cite{pitarakis2004least}).}.

\subsubsection{OLS-based Break-Point Date Estimator for a Single Structural Break}

\textbf{Step 1:} For any given $0 < \tau < 1$, denote with 
\begin{align}
\hat{\beta}_j ( \uptau) = \frac{ \displaystyle \sum_{t=1}^{[\uptau n] } y_t x_{t-1} }{  \displaystyle \sum_{t=1}^{[\uptau n] } x^2_{t-1} } \ \ \text{and} \ \ \hat{\beta}_3 ( \uptau) = \frac{ \displaystyle \sum_{t=[\uptau n]  + 1}^{n} y_t x_{t-1} }{  \displaystyle \sum_{t=[\uptau n]  + 1}^{n} x^2_{t-1} } 
\end{align}
Then the change-point estimator of $\uptau_2^0$ is defined as below
\begin{align}
\hat{\uptau}_{2,n} = \underset{ \uptau \in [0,1] }{ \text{argmin} } \ \text{RSS}_{2,n}(\uptau),
\end{align}
where 
\begin{align}
\text{RSS}_{2,n}( \uptau ) = \sum_{t=1}^{[\uptau n] } \bigg( y_t - \hat{\beta}_j ( \uptau) x_{t-1} \bigg)^2 +  \sum_{t=[\uptau n] + 1}^{n} \bigg( y_t - \hat{\beta}_3 ( \uptau) x_{t-1} \bigg)^2.
\end{align}

\newpage

Once we obtain $\hat{\uptau}_{2,n}$, the least squares estimator of $\beta_3$ is represented by $\hat{\beta}_3^{OLS} \left( \hat{\uptau}_{2,n} \right)$ and the OLS of $k_2^{0}$ is denoted by $\hat{k}_2 = [ \hat{\uptau}_{2,n} n ]$. 

\medskip
 
\textbf{Step 2:} For any given $0 < \tau < \hat{\tau}_{2,n}$, the OLS estimators of the parameters $\beta_1$ and $\beta_2$ are given by 
 \begin{align}
\hat{\beta}_1 ( \uptau) = \frac{ \sum_{t=1}^{[\uptau n] } y_t x_{t-1} }{  \sum_{t=1}^{[\uptau n] } x^2_{t-1} } \ \ \text{and} \ \ \hat{\beta}_2 ( \uptau) = \frac{ \sum_{t=[\uptau n]  + 1}^{ \hat{k}_2 } y_t x_{t-1} }{  \sum_{t=[\uptau n]  + 1}^{ \hat{k}_2 } x^2_{t-1} } 
\end{align}
respectively. Then the change-point estimator of $\uptau_1^0$ is defined as below 
\begin{align}
\hat{\uptau}_{1,n} = \underset{ \uptau \in [0, \hat{\uptau}_{2,n}  ] }{ \text{argmin} } \ \text{RSS}_{1,n}(\uptau),
\end{align}
where 
\begin{align}
\text{RSS}_{1,n}( \uptau ) = \sum_{t=1}^{[\uptau n] } \bigg( y_t - \hat{\beta}_1 ( \uptau) x_{t-1} \bigg)^2 +  \sum_{t=[\uptau n] + 1}^{\hat{k}_2} \bigg( y_t - \hat{\beta}_2 ( \uptau) x_{t-1} \bigg)^2.
\end{align}

Once we obtain $\hat{\uptau}_{1,n}$, the least squares estimator of $\beta_1$ and $\beta_2$ are represented by $\hat{\beta}_1^{OLS} \left( \hat{\uptau}_{1,n} \right)$ and $\hat{\beta}_2^{OLS} \left( \hat{\uptau}_{1,n} \right)$ respectively, and the OLS of $k_1^{0}$ is denoted by $\hat{k}_1 = [ \hat{\uptau}_{1,n} n ]$. 

\medskip

\begin{remark}
The break estimators that correspond to the OLS and IVX estimation, can be helpful to investigate which break-point will be identified first under abstract degree of persistence in predictive regression models. Therefore, after showing that the estimators $\hat{\uptau}^{OLS}_{1,n}$ and $\hat{\uptau}^{IVX}_{1,n}$ are consistent estimators of the true break-point $\pi_0$, we can determine their convergence rates and obtain useful performance statistics, under the alternative hypothesis, such as the average run length or the power loss function.  Thus, is important to examine also the convergence rates and asymptotic behaviour of the two estimators when under the alternative hypothesis of parameter instability, we obtain an estimator of the break point fraction and also modified parameter estimates which are functions of the estimated break point fraction instead of the initial sample size. 
\end{remark}

\subsubsection{IVX-based Break-Point Date Estimator for a Single Structural Break}

\textbf{Step 1:} For any given $0 < \tau < 1$, denote with 
\begin{align}
\hat{\beta}_j^{IVX} ( \uptau) = \frac{ \displaystyle \sum_{t=1}^{[\uptau n] } y_t \tilde{z}_{t-1} }{ \displaystyle  \sum_{t=1}^{[\uptau n] } \tilde{z}^2_{t-1} } \ \ \text{and} \ \ \hat{\beta}_3^{IVX}  ( \tau) = \frac{\displaystyle   \sum_{t=[\uptau n]  + 1}^{n} y_t \tilde{z}_{t-1} }{ \displaystyle   \sum_{t=[\uptau n]  + 1}^{n} \tilde{z}^2_{t-1} } 
\end{align}

\newpage

Then the change-point estimator of $\uptau_2^0$ is defined as below
\begin{align}
\hat{\uptau}_{2,n} = \underset{ \uptau \in [0,1] }{ \text{argmin} } \ \text{RSS}_{2,n}(\uptau),
\end{align}
where 
\begin{align}
\text{RSS}_{2,n}( \uptau ) = \sum_{t=1}^{[\uptau n] } \bigg( y_t - \hat{\beta}_j^{IVX} ( \uptau) \tilde{z}_{t-1} \bigg)^2 +  \sum_{t=[\uptau n] + 1}^{n} \bigg( y_t - \hat{\beta}_3^{IVX} ( \uptau) \tilde{z}_{t-1} \bigg)^2.
\end{align}
Once we obtain $\hat{\uptau}_{2,n}$, the IVX estimator of $\beta_3$ is represented by $\hat{\beta}_3^{IVX} \left( \hat{\uptau}_{2,n} \right)$ and the IVX of $k_2^{0}$ is denoted by $\hat{k}_2^{IVX} = [ \hat{\uptau}_{2,n} n ]$. 
 
\medskip

\textbf{Step 2:} For any given $0 < \tau < \hat{\tau}_{2,n}$, the IVX estimators of the parameters $\beta_1$ and $\beta_2$ are given by 
 \begin{align}
\hat{\beta}_1^{IVX} ( \uptau) = \frac{  \displaystyle \sum_{t=1}^{[\uptau n] } y_t \tilde{z}_{t-1} }{ \displaystyle \sum_{t=1}^{[\uptau n] } \tilde{z}^2_{t-1} } \ \ \text{and} \ \ \hat{\beta}_2^{IVX} ( \uptau) = \frac{ \displaystyle \sum_{t=[\uptau n]  + 1}^{ \hat{k}_2^{IVX} } y_t \tilde{z}_{t-1} }{ \displaystyle \sum_{t=[\uptau n]  + 1}^{ \hat{k}_2^{IVX} } \tilde{z}^2_{t-1} } 
\end{align}
respectively. Then the change-point estimator of $\uptau_1^0$ is defined as below 
\begin{align}
\hat{\uptau}_{1,n} = \underset{ \uptau \in [0, \hat{\uptau}_{2,n}  ] }{ \text{argmin} } \ \text{RSS}_{1,n}(\uptau),
\end{align}
\begin{align}
\text{RSS}_{1,n}( \uptau ) = \sum_{t=1}^{[\uptau n] } \bigg( y_t - \hat{\beta}_1^{IVX} ( \uptau) \tilde{z}_{t-1} \bigg)^2 +  \sum_{t=[\uptau n] + 1}^{\hat{k}_2} \bigg( y_t - \hat{\beta}_2^{IVX} ( \uptau) \tilde{z}_{t-1} \bigg)^2.
\end{align}
Once we obtain $\hat{\uptau}_{1,n}$, IVX estimator of $\beta_1$ and $\beta_2$ are represented by $\hat{\beta}_1^{IVX} \left( \hat{\uptau}_{1,n} \right)$ and $\hat{\beta}_2^{IVX} \left( \hat{\uptau}_{1,n} \right)$ respectively, and the IVX of $k_1^{0}$ is denoted by $\hat{k}_1^{IVX} = [ \hat{\uptau}_{1,n} n ]$. 

\medskip

\begin{remark}
Therefore, based on the break estimators above we want to investigate which break-point will be identified first under abstract degree of persistence using the two estimators. The break estimators that correspond to the OLS and IVX estimation, can be helpful to investigate which break-point will be identified first under abstract degree of persistence in predictive regression models. After showing that the estimators $\hat{\uptau}^{OLS}_{1,n}$ and $\hat{\uptau}^{IVX}_{1,n}$ are consistent estimators of the true break-point $\pi_0$, we can determine their convergence rates and obtain useful performance statistics, under the alternative hypothesis, such as the average run length or the power loss function.  
\end{remark}

\newpage 

\subsection{Main Aspects on Predictive Regression Model}
\label{sec2}

Therefore, notice that each of the covariates of the predictive regression are modelled using the autoregressive process $x_t = \rho_n x_{t-1} + v_t$, $x_0 = 0$. Specifically, when $\rho_n = \rho$ such that $| \rho | < 1$ then, $x_t$ is a stationary weakly dependent process. However, the present study focuses on cases where $x_t$ is nonstationary. In particular, in these cases $\rho_n = \left( 1 + \frac{c}{n}  \right)$.  More precisely, if the autoregressive parameter is fixed with $\rho_n = \rho$ and $| \rho | < 1$, then $x_t$ is asymptotically stationary and weakly dependent. Therefore, this framework helps to unify structural break testing in predictive regression models in cases where the properties of the predictor is not known thus offering robustness to integration order \citep{duffy2021estimation}.
\begin{align}
\label{model1}
y_{t}   &= \mu + \boldsymbol{\beta}^{\prime}  \boldsymbol{x}_{t-1} + u_{t}, \ \ \ \ \ \text{for} \ \ 1 \leq t \leq n,  
\\
\label{model2}
\boldsymbol{x}_t &= \left( \boldsymbol{I}_p - \frac{ \boldsymbol{C}_p }{ n^{ \gamma } } \right) \boldsymbol{x}_{t-1} + \boldsymbol{v}_{t}, 
\end{align} 
where $Y_{t} \in \mathbb{R}$ is an 1$-$dimensional vector and $\boldsymbol{x}_t \in \mathbb{R}^{p \times n}$ is a $p-$dimensional vector of local unit root regressors, with an initial condition $\boldsymbol{x}_0 = 0$. Moreover, $\boldsymbol{C} = \mathsf{diag} \{ c_1,...,c_p \}$ is a $p \times p$ diagonal matrix which determines the degree of persistence of the regressors by the unknown persistence coefficients $c_i$'s which are assumed to be positive constants. Define with $\boldsymbol{\eta}_t = \left( u_{t}, \boldsymbol{v}_{t}^{\prime} \right)^{\prime}$. Then, the partial sum process constructed from $\boldsymbol{\eta}_t $ satisfies a multivariate invariance principle. That is, for $r \in [0,1]$ and as $n \to \infty$ we have (where $\Rightarrow$ denotes weak convergence in distribution), 
\begin{align}
X_n(r) \Rightarrow n^{- 1/ 2} \sum_{t=1}^{ \floor{nr} } B(r)
\end{align}
where B(r) is a $p-$dimensional Brownian motion with covariance matrix 
\begin{align}
\boldsymbol{\Omega} 
=
\underset{ n \to \infty }{ \mathsf{lim} } \frac{1}{n} \mathbb{E} \left[ \left( \sum_{t=1}^n \boldsymbol{\eta}_t  \right) \left( \sum_{t=1}^n \boldsymbol{\eta}_t^{\prime} \right) \right] 
\equiv 
\underset{ n \to \infty }{ \mathsf{lim} } \frac{1}{n} \sum_{t=1}^n \sum_{j=1}^n \mathbb{E} \big[ \boldsymbol{\eta}_t \boldsymbol{\eta}_t^{\prime} \big]
\end{align}
Now $\boldsymbol{\Omega}$ and $\boldsymbol{B}(r)$ are partitioned as below
\begin{align}
\boldsymbol{\Omega}
&:= 
\begin{bmatrix}
\omega_{uu} & \omega_{vu}^{\prime}
\\
\omega_{uv} & \boldsymbol{\Omega}_{vv}
\end{bmatrix} = \boldsymbol{\Sigma} + \boldsymbol{\Lambda} +  \boldsymbol{\Lambda}^{\prime},
\ \ \ \
\boldsymbol{B}(r)
:=
\begin{bmatrix}
B_u(r)
\\
B_v(r)
\end{bmatrix}
\ \ \
\Omega_{\varepsilon} 
= 
\frac{1}{n} \sum_{t=1}^n \hat{\varepsilon}_t \hat{\varepsilon}_t^{\prime} + \frac{2}{n} \sum_{j=0}^n w \left( \frac{j}{M} \right) \sum_{t=j+1}^n \hat{\varepsilon}_t \hat{\varepsilon}_{t-j}^{\prime} ,
\end{align}
with partitions given by 
\begin{align}
\Sigma_{\epsilon} 
&= 
\frac{1}{n} \sum_{t=1}^n \hat{\varepsilon}_t \hat{\varepsilon}_t^{\prime}
\\
\Lambda_{\epsilon} 
&= 
\frac{1}{n} \sum_{j=0}^n w \left(  \frac{j}{M} \right) \sum_{t=j+1}^n \hat{\varepsilon}_t \hat{\varepsilon}_{t-j}^{\prime} 
\end{align}

\newpage 

\begin{remark}
Notice that due to the LUR specification of the autocorrelation matrix $\boldsymbol{R}_n$, the process $\boldsymbol{x}_t$ represents a restricted VAR model. When $\boldsymbol{x}_t$ is an unrestricted VAR model and the regressors of the predictive regression model has the same lag as the regressand then the system corresponds to the cointegrating predictive regression and thus different VAR representation theory is needed to handle the possible near-nonstationary and near-explosive components. We assume that the innovation sequence is a stationary vector and therefore we can examine the nonstationary properties of the predictors. However, it is important to emphasize that we refer to "disease-free equilibrium" (steady states) which imply the existence of such cointegrating relationships. 
\end{remark}

\begin{itemize}

\item In other words, although the limiting distributions are nonstandard and nonpivotal bootstrap-based simulation methodologies can be employed to obtain critical values or to obtain p-values. In particular, these limiting distributions are determined by sup functionals of Gaussian processes or  functionals which depends on the localizing coefficient of persistence. 

\item Furthermore, comparing the performance of the OLS with the IVX estimators as the the value of the localizing coefficient increases, that is, we move away from the unit boundary for low degrees of endogeneity (low values of the correlation between the innovation sequences) then both have about similar empirical size. However, as we move closer to the unity boundary, that is, for nearly-integrated regressors with high endogeneity we can clearly see that for the IVX estimator we obtain empirical size closer to the nominal size, while for the OLS estimator the empirical size is almost double or even three times the nominal size. In other words, the IVX estimators corrects the endogeneity bias under high persistence which we would get if we relied on inference based on the OLS estimator.  
    
\item  Note that in the case that $\omega_{uv} \neq 0$ the regressors are endogenous and in addition to regressor endogeneity the setting also allows for relatively unrestricted forms of serial correlation of the errors $\boldsymbol{\eta}_t$. These two aspects in general necessitate some form of modified least squares estimation in conjunction with HAC arguments to allow for the development of standard asymptotic inference.   

\item It is well known that in cointegrating regressions the OLS estimator is consistent despite the fact that the regressors are allowed to be endogenous and the errors are allowed to be serially correlated. However, the limiting distribution of the OLS estimator is contaminated by second order bias terms, reflecting the correlation structure between the regressors and the errors (see the paper: Tuning parameter free inference in Cointegrating Regressions). 

\item The literature provides several estimators which overcome this difficulty at the cost of tuning parameter choices such as: the number of leads and lags for the Dynamic OLS (D-OLS) estimator, kernel and bandwidth choice for the fully modified OLS estimator and the canonical cointegrating regression estimator. Such tuning parameters are often difficult to choose in practice and the finite sample performance of the estimators and tests based upon them often reacts sensitively to their choices.   

\item In contrast to the aforementioned approaches, the integrated modified OLS (IM-OLS) estimator avoids the choice of tuning parameters. However, standard asymptotic inference based on the IM-OLS estimator does require the estimation of a long-run variance parameter. Therefore, this is typically achieved by non-parametric kernel estimators, which necessitate kernel and bandwidth choices. In particular, to capture their effects in finite samples, V and Wagner (2014) propose fixed-b theory for obtaining critical values. However, their simulation results reveal that when endogeneity and/or error serial correlation is strong, a large sample size is needed for the procedure to yield reasonable sizes.  
    
\end{itemize}

\begin{wrap}

\color{blue}
In the literature various studies consider testing for stationarity especially for highly autocorrelated time series. Similarly, in the current study we investigate the size and power properties of structural break tests when the series are strongly dependent (high persistent) in a local-to-unity asymptotic framework. 

Standard tools in the literature for testing whether the largest autoregressive root is near the unit boundary  are unit root tests, which test the null hypothesis of a unit root against the alternative hypothesis of stationarity.

It is worth pointing out that the local-to-unity framework, developed by \cite{chan1987asymptotic1} and \cite{phillips1987towards}, generates relevant asymptotics for series whose dynamics are dominated by a large autoregressive root, that is, nearly integrated processes. Therefore, the particular asymptotic theory which is commonly employed when deriving the limiting distribution of test statistic and estimators of nonstationary time series models,  gives much more accurate approximations to small sample distributions compared to  standard asymptotics when the largest autoregressive root $\rho$ of a series is such that $n ( 1 - \rho )$ is smaller, than, say 30, where $n$ is the sample size. 
\color{black}

\color{magenta}
Therefore, in the current study we are interested to investigate the behaviour of structural break test statistics upon the estimator of the long-run variance in local-to-unity asymptotics.  According to \cite{muller2005size} more specifically, there are two key conjectures: First, estimators of the long-run variance that employ a bandwidth that goes to infinity more slowly that the sample size lead to tests of stationarity that reject even highly mean reverting series with probability one for a large enough sample size. Second, for some estimators of $\lambda$ that employ a bandwidth of the same order as the sample size, the resulting tests of stationarity do reject more often for less mean reverting series, but the exact properties depend crucially on which estimator $\hat{\lambda}$ is used.     

Generally speaking it is well understood that there cannot exist a statistic that perfectly discriminates between stationary and integrated processes in the local-to-unity framework - in fact, much of the appeal of this asymptotic device stems precisely from the fact that discrimination remains difficult even as the sample size increases without bound. 

\bigskip

Therefore, if we would compare the two proposed test statistics, that is, the sup OLS-Wald versus the sup IVX-Wald detectors for highly persistent data, that is, when the local-to-unity specification induces a value near the unit boundary then any conjectures regarding the performance of the tests based solely on this region of the parameter space, would most likely not give  us universally representative conclusions. In other words, the near observational equivalence between models with $\rho$ very close to unity and $\rho = 1$ makes it impossible to obtain the ideal asymptotic rejection profile. 

On the other hand, considering asymptotically optimal tests can provide statistical advantages. More precisely, point-optimal statistics, based on the Neyman-Pearson Lemma, optimally discriminate between a fixed level of mean reversion and no mean reversion. Thus, the asymptotic optimality property of the statistics requires Gaussian disturbances, but allows for unknown and very general correlations.



\medskip

\color{blue}
In particular in the paper of Cavaliere et al (2017) the authors show that the limiting null distributions of these test statistics are shown to be non-pivotal under heteroscedasticity, while tat of a robust Wald statistic, which is based around a sandwich estimator of the variance, is pivotal. Furthermore, they show that wild bootstrap implementations of the tests deliver asymptotically pivotal inference under the null. In particular, they demonstrate the consistency and asymptotic normality of the bootstrap estimators, and further establish the global consistency of the asymptotic and bootstrap tests under fixed alternatives. Monte Carlo simulations highlight significant improvements in finite sample behaviour using the bootstrap in both heteroscedastic and homoscedastic environments. 

\color{black}

\bigskip

\color{magenta}
In terms of the bootstrap procedure we apply the wild bootstrap to the corresponding unrestricted residuals which is the common practice in the bootstrap hypothesis testing literature. Moreover, in terms of testing methodology we consider the Wald test statistic under the assumption of homoscedastic Gaussian innovations, as wel as a heteroscedasticity-robust version of the Wald test implemented using a \cite{white1980heteroskedasticity} sandwich-type robust standard error. 
\color{black}

\paragraph{Algorithm 1}

The restricted Wild Bootstrap procedure is given as below

\begin{itemize}

\item[\textbf{Step 1.}]  Estimated the model using the Gaussian QML (unrestricted) which yield the estimate $\hat{\theta}$ together with the corresponding residuals such that $\hat{\varepsilon}_t := \varepsilon_t ( \hat{\theta} )$.

\

\item[\textbf{Step 2.}] Compute the centered residuals $\hat{\varepsilon}^{\mu}_t := \hat{\varepsilon}_t -  \frac{1}{n} \sum_{t=1}^n \varepsilon_t $ and construct the bootstrap errors such that $\hat{\varepsilon}_t^{*} = \hat{\varepsilon}^{\mu}_t \times w_t$, for $t = 1,...,n$. 

\

\item[\textbf{Step 3.}] Construct the bootstrap sample $\left\{ X_t^{*} \right\}$ from the bootstrap process
\begin{align}
X_t^{*}  = f( .,. ) u_t^{*}, \ \ \ t = 1,...,n.
\end{align}

\

\item[\textbf{Step 3.}] Using the bootstrap sample, $\left\{  X_t^{*} \right\}$, compute the bootstrap test statistic $S_n^{*}$ for testing the null hypothesis of linear restrictions. Define the corresponding $p-$value as $P_n^{*} := 1 - G_n^{*} ( S_n )$ with $G_n^{*} ( . )$ denoting the conditional (on the original data) cdf of $S_n^{*}$. Then, the wild bootstrap robust Wald statistic, denoted as $RW_t^{*}$, and associated $p-$value is defined analogously. Then, the wild bootstrap test of $H_0$ against $H_1$ at significance level $\alpha$ rejects the null hypothesis if $P_n^{*} \leq \alpha$. 

\end{itemize}

\color{red}
Notice that referring to the degree of persistence as a nuisance parameter might be nonstandard because $c_i$ is itself a random variable. Furthermore, the structural break setting of our paper assumes that the model coefficient remain constant over fixed intervals of time, therefore the use of the structural break modelling approach.     
\color{black}

\color{red}
In particular, when $\theta$ is a scalar parameter, the asymptotic bootstrap distribution of $\hat{\theta}^{*}$, is centred around $\hat{\theta}$ and bootstrap confidence intervals can be based on the empirical quantiles from the bootstrap distribution of $\hat{\theta}^{*}$, conditional on the original data. Specifically, letting $\hat{\theta}_{\alpha}^{*}$ denote the $\alpha \%$ quantile of the conditional bootstrap distribution of $\hat{\theta}^{*}$, then the asymptotic $( 1 - \alpha ) \%-$level naive (or basic) and percentile bootstrap confidence intervals are given by 
\begin{align}
\left[ 2 \hat{\theta} - \hat{\theta}_{( 1 - \frac{\alpha}{2}   )}^{*} ; 2 \hat{\theta} - \hat{\theta}_{(\frac{\alpha}{2})}^{*} \right] \ \ \ \text{and} \ \ \ \left[ \hat{\theta}_{( 1 - \frac{\alpha}{2}   )}^{*} ; \hat{\theta}_{(\frac{\alpha}{2})}^{*} \right]
\end{align}  
respectively. Alternatively, the studentized bootstrap confidence interval can be constructed from the associated $t-$statistics of the bootstrap estimates. 
\color{black}

\bigskip

\color{blue}
If $x_t$ and $y_t$ are both $I(0)$ and if $x_t$ and $\epsilon_{2t}$ are correlated, any OLS estimators, do not yield consistent estimation of the the slope parameter. More specifically, the OLS estimator is not robust to integration orders when the error is correlated. In particular, it can be said that the co-integration among $I(0)$ variables extends the applicability of OLS beyond that in the $I(0)$ variables by freeing us from worry about the correlation between the regressor and disturbances. Furthermore, the instrumental variable estimator has been used for the regression model with a correlated error when the variables are $I(0)$. Therefore, the instrumental variable estimator applied to the cointegrated regression with a correlated error is $n-$consistent, but some modifications are required to make it mixed Gaussian. 

For example, consider the model $y_t = \boldsymbol{b}^{\prime} \boldsymbol{x}_t + \varepsilon{\epsilon}_{2t}$ with a $k-$element vector $\boldsymbol{x}_t$, when $\left\{ \boldsymbol{x}_t \right\}$ is co-integrated.    

\color{black}

In particular, the literature for the case of local-to-unity specification focuses on the endogeneity issue.  
\begin{example}
\begin{align}
y_t &= \mu_y + \beta x_{t-1} + u_t
\\
x_t &= \mu_x + \rho x_{t-1} + v_t
\end{align}

\medskip

The innovation sequence $\eta_t := \big( u_t, v_t \big)^{\prime}$ is assumed to form a (bivariate) martingale difference sequence with respect to the natural filtration $\mathcal{F}_t = \sigma \big( \eta_t, \eta_{t-1},...  \big)$ with covariance matrix given by the following expression 
\begin{align}
\mathbb{E} = \left( \eta_t \eta_t^{\prime} \| \mathcal{F}_{t-1} \right)
=
\begin{bmatrix}
\sigma_x^2 & \sigma_{xy} 
\\
\sigma_x^2 & \sigma_{y}^2  
\end{bmatrix}
\end{align}
Then the $t-$statistic is given by 
\begin{align}
\mathcal{n} :=  \frac{ \displaystyle  \sum_{t=1}^n \big( x_{t-1} - \hat{\mu}_x \big) \big( y_t - \hat{\mu}_y  \big) }{ \displaystyle  \sqrt{ \hat{\sigma}_u^2  \sum_{t=1}^n \big( x_{t-1} - \hat{\mu}_x \big)^2 }  }  
\end{align} 
\end{example}

\medskip

\color{red}
Furthermore, we can consider the dating of the explosive episode such that for $\hat{\tau}_e = \floor{ n \hat{r}_e }$ with $\hat{r}_e$ being equal to either 
\begin{align}
\hat{r}^{ \rho }_e = \underset{ s \geq r_0 }{ \mathsf{inf} } \big\{ s: DF_s^{\rho} > c \nu_{ \theta_n }^{ \rho } (s) \big\} \ \ \text{or} \ \ \ \hat{r}^{ t }_e = \underset{ s \geq r_0 }{ \mathsf{inf} } \big\{ s: DF_s^t > c \nu_{ \theta_n }^t (s) \big\} 
\end{align}
where $\nu_{ \theta_n }^{ \rho } (s)$ and $\nu_{ \theta_n }^t (s)$ denote the right hand side 100 $ \theta_n \%$ critical values of two statistics based on $\lambda_s = \floor{ ns }$ observations and $\theta_n$ is the size of the one-sided test. 
\color{black}

\end{wrap}

\section{Background Main Results}   
   
\begin{lemma}
Consider the following predictive regression model 
\begin{align}
\label{model1}
y_t 
= 
\begin{cases}
\beta_1 y_{t-1} + \epsilon_t, & 1 \leq t \leq k_1^0, \\
\beta_2 y_{t-1} + \epsilon_t, & k_1^0 + 1 \leq t \leq k_2^0,  \\
\beta_3 y_{t-1} + \epsilon_t, & k_2^0 + 1 \leq t \leq n
\end{cases}
\end{align}
Under Assumptions C1-C4, the following results hold jointly
\begin{align}
\frac{1}{n} \sum_{t=1}^{ k_1^0 } y_{t-1} u_t 
\Rightarrow 
\frac{\sigma^2 }{2} \left( W^2 ( \tau_1^0 ) - \tau_1^0 \right),
\ \ \
\frac{1}{n^2} \sum_{t=1}^{ k_1^0 } y_{t-1}^2 
\Rightarrow 
\sigma^2 \int_0^{ \tau_1^0 } W^2 (s) ds, 
\ \ \ 
\frac{ y_{ k_1^0 } }{ \sqrt{n} } 
\Rightarrow \sigma W \left(  \tau_1^0 \right)
\end{align}
\end{lemma}   

\begin{proof}
To prove part (a), note that 
\begin{align}
y_t = y_0 + \sum_{ j=1 }^t u_j = \frac{tc}{ n^{\eta} } + y_0 + \sum_{ j=1 }^t \epsilon_j, 0 \leq t \leq k_1^0, 
\end{align}
Therefore, we obtain that 
\begin{align}
\frac{1}{n} \sum_{ j=1 }^{ k_1^0 }  y_{ t-1 } u_t 
= 
\frac{1}{n} \sum_{ j=1 }^{ k_1^0 } \left( \frac{ (t-1)c }{ n^{\eta} } + y_0 + \sum_{ j=1 }^{ t } \epsilon_j \right)  \left( \frac{c}{ n^{\eta} } + \epsilon_t \right) 
= 
\frac{1}{n} \sum_{ j=1 }^{ k_1^0 } \left(  \sum_{ i=1 }^{ t - 1 } \epsilon_t \right) \epsilon_t + o_p(1),
\end{align}
Since $\eta > 1/2$ and by applying standard results in the unit root literature, we have that 
\begin{align}
\frac{1}{n} \sum_{ j=1 }^{ k_1^0 } y_{ t-1 } u_t \Rightarrow  \frac{\sigma^2 }{2} \int_0^{ \tau_1^0 } W(s) dW(s)  =  \frac{\sigma^2 }{2} \left( W^2 ( \tau_1^0 ) - \tau_1^0 \right)
\end{align} 
\end{proof}                                                                                                       

\newpage

Furthermore, note that it can be proved that 
\begin{align}
\frac{ \displaystyle \sum_{ t=1}^{ \hat{k}_1^0 } y_{t-1}^2 }{ \displaystyle \sum_{ t=1}^{ k_1^0 } y_{t-1}^2 } \overset{ p }{ \to } 1, \ \ \ \frac{ \displaystyle \sum_{ t= \hat{k}_1^0 + 1 }^{ \hat{k}_2^0 } y_{t-1}^2 }{ \displaystyle \sum_{ t= k_1^0 + 1 }^{ k_2^0 } y_{t-1}^2 } \overset{ p }{ \to } 1, \ \ \ \frac{ \displaystyle \sum_{ t= \hat{k}_2^0 + 1 }^{ n } y_{t-1}^2 }{ \displaystyle \sum_{ t= k_2^0 + 1 }^{ n } y_{t-1}^2 } \overset{ p }{ \to } 1 \ \ \ \text{and} \ \ \ \frac{ \hat{ \sigma }^2 }{ \sigma^2 } \overset{ p }{ \to } 1. 
\end{align}
Another important aspect is to consider the case of weakly dependent errors which applies removing the independence assumption of the error sequence $\left\{ \epsilon_t \right\}_{ t = 1}^n$. Specifically, the case of weakly dependent errors assumes the following linear process representation 
\begin{align}
\epsilon_t = \sum_{ j = 0}^{ \infty } a_j e_{t-j}, \ \ \text{for} \ \ t \geq 1. 
\end{align} 
Furthermore, to ensure that $\epsilon_t$'s are weakly dependent error terms we assume that $a(1) := \sum_{ j = 0}^{ \infty } a_j \neq 0$, $\sum_{ j = 0}^{ \infty } j^{3 / 2} | a_j | < \infty$, and $\left\{ e_t \right\}_{ t = 1}^n$ is a sequence of \textit{i.i.d} random variables with mean zero and variance $0 < \sigma_e^2 < \infty$. Moreover, the consistency of three estimators that correspond to the three regimes, when the estimated break is incorporated, needs to be established such that (Theorem 4.1 in \cite{pang2021estimating})
\begin{align}
\begin{cases}
\tau_1^0 n \left( \hat{\beta}_1 \left( \hat{\tau}_{1,n} \right) - \beta_1 \right) \Rightarrow \frac{ \displaystyle W^2(1) - a(2)/ a^2 (1) }{ \displaystyle 2 \int_0^1 W^2 (s) ds}, 
\\
\\
\sqrt{ \frac{ \tau_1^0 n^{ 1 + \alpha_1 } }{ 2c_1 }  } \beta_2^{ k_2^0 - k_1^0 } \left( \hat{\beta}_2 \left( \hat{\tau}_{1,n} \right) - \beta_2 \right) \Rightarrow \xi, 
\\
\\
\sqrt{ \frac{ \tau_1^0 n^{ 1 + \alpha_2 } }{ 2c_2 }  } \beta_2^{ k_2^0 - k_1^0 } \left( \hat{\beta}_3 \left( \hat{\tau}_{2,n} \right) - \beta_3 \right) \Rightarrow \zeta, 
\end{cases}
\end{align}
where $a(2) = \sum_{ j = 0 }^{ \infty } a^2_j$, and $\xi$ and $\zeta$ are two independent standard Cauchy variates. 

\begin{remark}
In terms of the identification of the break point this is determined by the stochastic orders of $P_1$ and $P_2$. If $P_2$ has a higher stochastic order of magnitude that $P_1$, then $RSS( \tau_1^0 ) - RSS( \tau_2^0 )$ will diverge to $\infty$ in probability, and $k_2^0$ will be identified first asymptotically. Instead, if $P_1$ has a higher stochastic order of magnitude than $P_2$, $RSS( \tau_1^0 ) - RSS( \tau_2^0 )$ will go to $- \infty$ in probability, and $k_1^0$ will be identified first asymptotically. However, when $P_1$ and $P_2$ have the same stochastic order of magnitude, which break will be identified first depends on the magnitude and the duration of the break, which are unobservable in reality. Therefore, it is difficult to determine which break will be uncovered first, and we need to test and estimate the second break from all subsamples split by the first estimated break point \citep{pang2021estimating}.
\end{remark}

\newpage 

\underline{Local to Unity}

Consider the following predictive regression model 
\begin{align}
\label{model1}
y_t 
= 
\begin{cases}
\beta_1 y_{t-1} + u_t, & 1 \leq t \leq k_1^0, \\
\beta_2 y_{t-1} + u_t, & k_1^0 + 1 \leq t \leq k_2^0,  \\
\beta_3 y_{t-1} + u_t, & k_2^0 + 1 \leq t \leq n
\end{cases}
\end{align}
where $\beta_1 = \left( 1 - \frac{\gamma}{n} \right)$ with $\gamma \in \mathbb{R}$, $\beta_2 = \beta_{2n} = \left( 1 + \frac{c_1}{k_n} \right)$ with $c_1 > 0$, $\beta_3 = \beta_{3n} = \left( 1 - \frac{c_2}{h_n} \right)$ with $c_2 > 0$ and $u_t = cn^{- \eta} + \epsilon_t$ with $c \in \mathbb{R}$ and $\eta > 1 /2$ when $t \leq k_1^0$, and $u_t = \epsilon_t$ when $t > k_1^0$.  

Then $\hat{\beta}_1 \left(  \hat{\tau}_{1,n} \right)$, $\hat{\beta}_2 \left(  \hat{\tau}_{1,n} \right)$ and $\hat{\beta}_3 \left(  \hat{\tau}_{2,n} \right)$ are all consistent, and their asymptotic distributions are respectively given by 
\begin{align}
\begin{cases}
n \left( \hat{\beta}_1 \left( \hat{\tau}_{1,n} \right) - \beta_1 \right) \Rightarrow \frac{ \displaystyle \int_0^{ \tau_1^0 } \int_0^r e^{\gamma(r-s)} dW(s) dW(r) }{ \displaystyle \int_0^{ \tau_1^0 } \left(  \int_0^r e^{\gamma(r-s)} dW(s) \right)^2 dr}, 
\\
\\
\sqrt{ \displaystyle \frac{ n k_n }{ 2c_1 }  } \beta_2^{ k_2^0 - k_1^0 } \left( \hat{\beta}_2 \left( \hat{\tau}_{1,n} \right) - \beta_2 \right) \Rightarrow \frac{ \displaystyle \sqrt{2 c_1 } \mathcal{X} }{ \displaystyle \int_0^r e^{\gamma(\tau_1^0 -s)} dW(s) }, 
\\
\\
\sqrt{ \displaystyle \frac{ n h_n }{ 2c_2 }  } \beta_2^{ k_2^0 - k_1^0 } \left( \hat{\beta}_3 \left( \hat{\tau}_{2,n} \right) - \beta_3 \right) \Rightarrow \frac{ \displaystyle \sqrt{2 c_2 } \mathcal{Z} }{ \displaystyle \int_0^{ \tau_1^0 } e^{\gamma(\tau_1^0 -s)} dW(s) },  
\end{cases}
\end{align}
where $\mathcal{X}$ and $\mathcal{Z}$ are two random variables obeying $\mathcal{N} \left( 0, \frac{1}{2c_1} \right)$ and $\mathcal{N} \left( 0, \frac{1}{2c_2} \right)$ respectively. Moreover, the random variables $\mathcal{X}$, $\mathcal{Z}$ and $\left\{ W(s), 0 \leq s \leq \tau_1^0 \right\}$ are mutually independent.  

\medskip

\underline{Additional Hypothesis Testing as in \cite{pang2021estimating}:}

We discuss a hypothesis test problem which is important for testing for the equality of the persistence in the two shifting regimes. Suppose that $k_n = n^{ \alpha_1 }$ with $0 < \alpha_1 < 1$ and $h_n = n^{ \alpha_2 }$ with $0 < \alpha_2 < 1$. Therefore, the null hypothesis is given by $H_0: c_1 = c_2 = 0$.
Under the null, this implies that the underlying stochastic process has an autoregressive unit root throughout the sample. The alternative hypothesis can be formulated as below $H_0: c_1 \neq 0$ and $c_2 \neq 0$. Therefore, under the alternative hypothesis the underlying process is generated by the AR(1) model with $k_n = n^{ \alpha_1 }$ and $h_n = n^{ \alpha_2 }$. Denote with
$\hat{ \beta } = \frac{ \sum_{t=1}^n y_t y_{t-1} }{ \sum_{t=1}^n y^2_{t-1}  }$, and the associated test statistic is given by 
\begin{align}
t_n = \sqrt{ \sum_{t=1}^n y^2_{t-1} } \left( \hat{ \beta } - 1 \right) = \frac{ \displaystyle \sum_{t=1}^n y_t y_{t-1} }{ \displaystyle  \sqrt{ \sum_{t=1}^n y^2_{t-1} } }.   
\end{align}

\newpage 

Furthermore, it is well known that under $H_0$, we have that 
\begin{align}
t_n \Rightarrow \frac{ W^2(1) - 1 }{ 2 \sqrt{ \int_0^1 W^2 (s) ds } }
\end{align}
which implies that the the t-ratio given by the random variable $t_n$ is bounded in probability under $H_0$. On the other hand, since we can prove that $| t_n |$ will go to infinity in probability under $H_1$, it implies that the t-statistic has the ability to disciminate between the data generated from a unit root model vis-a-vis the date generated from the model specification above. 

\medskip

\begin{example}[Structural Breaks both in Mean and Variance]
Consider the following DGP
\begin{align}
y_t &= \boldsymbol{x}_{t-1}^{\prime} \boldsymbol{\beta}_t + u_t, \ \ \ \ \ \ \text{for} \ \ t = 1,..., n,
\\
\boldsymbol{x}_t &= \left( \boldsymbol{I}_p - \frac{ \boldsymbol{C}_p }{ n^{ \gamma } } \right) \boldsymbol{x}_{t-1} + \boldsymbol{v}_{t},
\end{align}
We consider a time-varying coefficient vector written as 
\begin{align}
\boldsymbol{\beta}_t = \boldsymbol{\beta}_1 \mathbf{1} \big( t \leq k_1^0  \big) + \boldsymbol{\beta}_2 \mathbf{1} \big( t > k_1^0  \big) 
\end{align}
which implies a structural break setting in the coefficients of the predictive regression model. Notice for the moment, we assume that the model of interest does not have an intercept. In other words, we do not consider a transformation of the $y_t$ random variable with a model intercept. Furthermore, we introduce conditional heteroscedastity for the innovation term of the predictive regression model such that $u_t = \sigma_t \epsilon_t$ and we consider the existence of a possible break-point in the conditional variance at an unknown location which can be at a different sample fraction from the break-point that corresponds to the model parameters. Therefore, it follows that
\begin{align}
\sigma_t = \sigma_1 \mathbf{1} \big( t \leq k_2^0  \big) + \sigma_2 \mathbf{1} \big( t > k_2^0  \big) 
\end{align}  
where $\sigma_1 > 0$ and $\sigma_2 > 0$ and the structural break occurring in its variance at some unknown location $k_2^0$. In other words, within the particular testing environment we consider the case of a break in the conditional mean parameters of the predictive regression model and error variance occurring at break-point locations  $k_1^0$ and $k_2^0$ respectively. For simplicity we will refer to the break in the model coefficients, that is, the parameter vector $\boldsymbol{\beta}_i$ and $\sigma_i$ for $i \in \left\{ 1, 2 \right\}$ as the break in mean and variance respectively.

\end{example}

\newpage 

\section{Conclusion}


Our main research objective is to investigate the statistical properties and asymptotic behaviour of break-point estimators when a single change-point occurs in predictive regression models.

In this article we consider testing for structural break in univariate time series regressions models, an active research literature considers statistical methodologies for estimation and inference in change-point models for multivariate time series (see, \cite{preuss2015detection} among others). We leave such considerations and extensions of our  framework as future research. Other extensions include consider a sequence of break-point estimators in a high-dimensional predictive regression model using shrinkage type estimators to obtain the break-point date estimators (see, the excellent framework proposed recently by \cite{tu2023penetrating}\footnote{Although the term "sporadic predictability", roughly speaking reflects predictability adapted to the changing environment based on the myriads of global and local macroeconomic shocks not to mention that the first word of the particular paper title is mildly not appropriate for an economics paper.}, see also the discussion on the shrinkage inference approach in \cite{katsouris2023quantile} and the proposed shrinkage methodology presented in \cite{chen2017estimation}, \cite{nkurunziza2010shrinkage} as well as \cite{nkurunziza2021inference}). 

\begin{itemize}

\item Predictive regression models are employed for statistical inference purposes when the lagged value of a financial variable is used as predictor for next-period stock returns. Therefore, using the proposed Wald-type statistics for detecting parameter instability, we examine whether stock returns are predictable even when we account for the presence of a single structural break in the relation between the regressand and the persistent predictors. Although the predictive regression errors are uncorrelated to any predetermined regressor are allowed to be contemporaneously correlated with the innovations of the local unit root processes. 

\item The asymptotic theory is established using weak convergence arguments for sample moment functionals into OU functionals (see, \cite{uhlenbeck1930theory}). In order to correct for the finite sample effects of estimating the model intercept, which are most pronounced for highly persistent regressors that are strongly correlated with the predictive model's innovations, \cite{kostakis2015robust} recommend the use of a finite-sample correction. On the other hand, the inclusion of this correction factor does not alter any of the large sample results.

\item 
In other words, generally when we consider the joint predictability and parameter instability tests we find that a rejection by these statistics can only be interpreted as signaling instability in any of the model parameters. Therefore, if the objective is not only to test for overall model stability but also to determine which particular subset of parameters is unstable a sup-Wald of joint stability of all parameters in the predictive regression model is needed.      
\end{itemize}

\newpage 

\section{Appendix}

\subsection{Auxiliary Results}

\begin{lemma}
Assume that $x_t$ is generated by 
\begin{align}
x_t = \left( 1 - \frac{c}{n^{\gamma}} \right) x_{t-1} + u_t, \ \ x_0 = 0, c > 0, \gamma \in (0,1). 
\end{align} 
with $t \in \left\{ 1,..., n \right\}$ where $u_t \sim_{\textit{i.i.d}} \mathcal{N}( 0, \sigma_u^2 )$. Then the following results hold:
\begin{align}
(i) \ \ &\frac{1}{ n^{ 3 / 2 } } \sum_{t=1}^{ \floor{ns} }  x_{t-1} \overset{ d }{ \to } \int_0^{s} J_c(r) dr
\\
(ii) \ \ &\frac{1}{ n^{ \frac{1}{2} + \gamma_x }} \sum_{t=1}^{ \floor{ns} } x_{t-1} \overset{ p }{ \to } J_c( s )
\end{align}
\end{lemma}

\begin{example}
Consider the integrated process: $ (1 - L) x_t = u_t$. Then, by recursive substitution we obtain that $x_t = \sum_{ j=1 }^t u_j + x_0$, for $1 \leq  t \leq n$. Define with 
\begin{align}
X_n(r) = \frac{1}{ \sqrt{n} } \sum_{ j=1 }^{ \floor{nr} } u_j = \frac{1}{ \sqrt{n} } S_{ \floor{nr} } \in \mathcal{D} (0,1). 
\end{align}
\begin{align*}
\sum_{ j=1 }^n x_t 
&= 
\sum_{ j=1 }^n \big[ S_{j-1} + u_j + x_0  \big] 
= \sqrt{n} \sum_{ j=1 }^n \left[ \frac{1}{ \sqrt{n} }  S_{j-1}  \right] + \sum_{ j=1 }^n u_j + x_0
\\
&= n \sqrt{n} \sum_{ j=1 }^n \left[ \int_{ (j-1) / n}^{ j / n} X_n(r) dr \right] + \sum_{ j=1 }^n u_j + x_0
\\
&= n^{3/2} \int_0^1 X_n(r) dr  + \sum_{ j=1 }^n u_j + x_0
\end{align*}
Hence, 
\begin{align*}
\frac{1}{ n^{3 / 2} }   \sum_{ j=1 }^n  
&= 
\int_0^1 X_n(r) dr  + \sum_{ j=1 }^n u_j + x_0
= 
\int_0^1 X_n(r) dr + o_P(1)
\\
&\Rightarrow \int_0^1 B(r) dr  \ \ \ \text{where} \ \ B(r) \equiv BM ( \omega^2 ). 
\end{align*}
Therefore, it holds that 
\begin{align}
\frac{1}{ n^{3 / 2} }  \sum_{ j=1 }^{ \floor{ns} } x_t \Rightarrow \int_0^r B(s) ds.  
\end{align}
\end{example}

\newpage

\begin{example}
\end{example}
Consider the LUR autoregression
\begin{align}
x_t = \rho_n x_{t-1} + u_t, \ \ x_0 = 0, c > 0, \gamma_x \in (0,1), 
\end{align} 
where $\rho_n := \left( 1 - \frac{c}{n^{\gamma}} \right)$. Then, we obtain that 
\begin{align}
x_{t-1} = \rho_n^{t-1} \sum_{j=1}^{ t - 1} \rho_n^{-j} u_j \equiv   \left( 1 - \frac{c}{n^{\gamma}} \right)^{t-1} \sum_{j=1}^{ t - 1} \left( 1 - \frac{c}{n^{\gamma}} \right)^{-j} u_j. 
\end{align}
where $c > 0$ and $\gamma \in (0,1)$. Thus, we need to show that 
\begin{align}
\label{the.exp1}
\frac{1}{ n^{ 3 / 2 } } \sum_{t=1}^{ \floor{ns} }  x_{t-1} \overset{ d }{ \to } \int_0^{s} J_c(r) dr
\end{align}
\begin{proof}
We have that
\begin{align*}
\label{the.exp2}
\sum_{t=1}^{ n } x_{t-1} 
&= 
\sum_{t=1}^{ n } \left\{ \left( 1 - \frac{c}{n^{\gamma}} \right)^{t-1} \sum_{j=1}^{ t - 1} \left( 1 - \frac{c}{n^{\gamma}} \right)^{-j} u_j \right\}
= 
\sum_{t=1}^{ n } \left\{ \left( 1 - \frac{c}{n^{\gamma}} \right)^{-1} \sum_{j=1}^{ t - 1} \left( 1 - \frac{c}{n^{\gamma}} \right)^{t-j} u_j \right\}
\\
&=
\left( 1 - \frac{c}{n^{\gamma}} \right)^{-1} \sum_{t=1}^{ n } \left\{  \sum_{j=1}^{ t - 1} \left( 1 - \frac{c}{n^{\gamma}} \right)^{t-j} u_j \right\}
\end{align*}
Next, consider the boundaries on the above expression such that $t \in \left\{ 1,..., \floor{ ns } \right\}$, which gives
\begin{align}
\frac{1}{ n } \left( 1 - \frac{c}{n^{\gamma}} \right)^{-1}  \frac{1}{ \sqrt{n} } \sum_{t=1}^{  \floor{ ns } } \left\{  \sum_{j=1}^{ t - 1} \left( 1 - \frac{c}{n^{\gamma}} \right)^{t-j} u_j \right\}, \ \ \text{for} \ s \in [0,1],
\end{align}
Therefore as $n \to \infty$ we obtain the limit: 
\begin{align}
\frac{1}{ n^{ 3 / 2 } } \sum_{t=1}^{ \floor{ns} }  x_{t-1} \overset{ d }{ \to } \int_0^{s} J_c(r) dr
\end{align}
\end{proof}

\begin{example}[see, \cite{kejriwal2020bootstrap}]
Consider testing the null hypothesis, $H_0: \alpha_i = 1$, for all $i$, then regression model can be written as 
\begin{align*}
\Delta y_t = c_i +  ( \alpha_i - 1 ) y_{t-1} + \sum_{j=1}^{p-1} \uptau_j \Delta y_{t-j} + e_t^{*}, \ \ \ c_i = ( 1 - \alpha_i ) \big[  u_{ n_{i-1 }}^0  + \mu_i \big].    
\end{align*}
\end{example}

\newpage 

\begin{wrap}

\subsubsection{Pivot Transformation}

\color{red}
\begin{definition}[Highest density region]
Let the density function $f_Y(y)$ of a random variable $Y$ defined on a probability space $\left( \Omega_Y, \mathcal{F}_{Y}, \mathbb{P}_{Y} \right)$ and taking values on the measurable space $( \mathcal{Y}, \mathcal{M} )$ be continuous and bounded. Then, the $(1 - \alpha) 100 \%$ HDR is a subset $\mathbf{S} ( k_{\alpha} )$  of $\mathcal{Y}$ defined as $\mathbf{S} ( k_{\alpha} ) := \big\{ y: f_Y(y) \geq k_{\alpha} \big\}$, where $k_{\alpha}$ is the largest constant that satisfies $\mathbb{P}_{Y} \big(  Y \in  \mathbf{S} ( k_{\alpha} ) \big) \geq 1 - \alpha$.   
\end{definition}

\color{blue}

Consider a pivot transformation via the following bounded linear operator. For any bounded function $x: [ 0,1 ] \mapsto \mathbb{R}$ and a number $a \in [0,1]$ such that $x(a) \neq 0$, define the operator $\tilde{\Delta} [x, a]$ on functions with domain $[0,1]$ by 
\begin{align}
\tilde{\Delta} [x, a] y(s) = y(s) - \frac{ x(s) }{ x(a) } y(a).  
\end{align}
Notice that if $\bar{u}_n = O_p(1)$ and $x(s) = s$ then
\begin{align*}
\sqrt{n} \mathcal{B}_n(s) 
&= 
\frac{ 1 }{ \sqrt{n} } \sum_{ t= 1}^{ \floor{ns} } \left( u_n - \bar{u}_n \right)
\\
&=
\frac{ 1 }{ \sqrt{n} } \sum_{ t= 1}^{ \floor{ns} } u_T - \sqrt{n} \frac{ \floor{ns} }{ n } \bar{u}_n
\\
&=
\sqrt{n} \mathcal{U}_n (s) - s \sqrt{n} \mathcal{U}_n (1) + o_p(1)
\\
&= 
\tilde{\Delta} [x, 1] \sqrt{n} \mathcal{U}_n (s) + o_p(1), 
\end{align*}

Moreover, we define the compounded operator $\tilde{\Delta} [ \underline{ \textsl{g} } _n^{(n)} , \underline{a}_n ]$ with $\underline{ \textsl{g} } _n^{(n)} := \left( \textsl{g}^{(1)},...,   \textsl{g}^{(n)} \right)$ and $\underline{a}_n := \left( a_1,..., a_n \right) \in [0,1]^n$ where $n \in \mathbb{N}$, via the following expression  
\begin{align}
\tilde{\Delta} [ \underline{ \textsl{g} } _n^{(n)} , \underline{a}_n ] := \tilde{\Delta} [ \underline{ \textsl{g} } _n^{(1)} , \underline{a}_1 ] \circ \tilde{\Delta} [ \underline{ \textsl{g} } _n^{(2)} , \underline{a}_2 ] \circ \dotsb \circ \tilde{\Delta} [ \underline{ \textsl{g} } _n^{(n)} , \underline{a}_n ], 
\end{align} 
where $\textsl{g}_k^{(k)}$ is defined recursively through a starting set of functions $\textsl{g}_1(s),...,\textsl{g}_n(s)$. The recursion is given by 
\begin{align}
\label{rec}
\textsl{g}_j^{(n)} (s) &= \textsl{g}_j^{(s)}
\\
\textsl{g}_j^{( n - k -1 )} (s) &:= \textsl{g}_j^{(n-k)} (s) - \frac{ \textsl{g}_{ n - k }^{(n - k) } (s) }{ \textsl{g}_{ n - k }^{(n - k) } ( a_{n-k} ) } \textsl{g}_j^{ (n-k) } (  a_{ n - k } )  
\end{align}

for $k \in [ 0, n-2 ]$ and $j \in [ 1, n - k ]$. Then the following lemma shows that $\tilde{\Delta} [ \underline{ \textsl{g} } _n^{(n)} , \underline{a}_n ]$ is a transformation that has good properties and that cancels out the asymptotic gap between $\sqrt{n} \hat{ \mathcal{U} }_n$ and $\sqrt{n} \mathcal{U}_n$. 

\begin{lemma}
(Algorithm to reach a pivot statistic). Let $\left( a_1, a_2,..., a_n \right) \in [0,1]^n$ be real numbers with $n \in \mathbb{N}$. Let $\textsl{g}_1(s), \textsl{g}_2(s),..., \textsl{g}_n(s)$ be a set of known bounded real-valued functions with domain $[0,1]$ \textit{s.t}, for all $j \in [1,n]$, $\textsl{g}_1(s)^{(j)}( a_j  ) \neq 0$ where $\textsl{g}_j(j)$ are defined by recursion \eqref{rec}. Then, for any function $f : [0,1] \to \mathbb{R}$ of the form $f(s) = \sum_{j =1 }^n b_j \textsl{g}_j (s)$:   
\begin{enumerate}
\item[(i)] $\tilde{\Delta} [ \underline{ \textsl{g} } _n^{(n)} , \underline{a}_n ] f(s) = 0$, and

\item[(ii)] if, for all $j \in [1,n]$, sup$_{ s \in [0,1] } \left|  \textsl{g}_j(s) \right| < \infty$, then $\tilde{\Delta} [ \underline{ \textsl{g} } _n^{(n)} , \underline{a}_n ]$ is a linear bounded operator on any linear subspace of the space of real-valued functions with domain $[0,1]$, and thus it is continuous on the same linear subspace.   
\end{enumerate}
\end{lemma}

\color{black}

\subsection{Globally Convergent Algorithm for Wald test}

Three commonly used multivariate tests based upon the maximization of the log-likelihood function are the Wald (W), Likelihood Ratio (LR) and the Lagrange Multiplier (LM) Tests.

We define with $\theta = \left( \mu_1, \mu_2, \Sigma_1, \Sigma_2 \right)$. Furthermore, for a certain hypothesized restriction on the parameter space of means
\begin{align}
H_0: c( \mu_1, \mu_2 ) = q,  
\end{align} 

Moreover, define with $\hat{\theta}$ the unrestricted MLE of $\theta$, and let $\hat{\theta}_r$ denote the MLE under the restriction $H_0:$, that is, the solution to the problem
\begin{align}
\underset{ \theta }{ \mathsf{max} } \ \ \ell ( \theta ) \ \ \text{subject to} \ \ c( \mu_1, \mu_2 ) = q.  
\end{align} 

The test statistics of interest are defined as 
\begin{align}
W = \bigg[ c( \hat{\mu}_1, \hat{\mu}_2 ) - q \bigg]^{\prime} \bigg( \text{Var} \big( c( \hat{\mu}_1, \hat{\mu}_2 ) - q  \big) \bigg)^{-1} \bigg[ c( \hat{\mu}_1, \hat{\mu}_2 ) - q \bigg]
\end{align}

\subsection{A robust test for predictability}

We are interested in testing the following hypothesis:
\begin{align}
H_0: \mathsf{Cov} \big( y_t, \boldsymbol{x}_{t-1} \big) = 0 \ \ \ \text{against} \ \ \  \mathsf{Cov} \big( y_t, \boldsymbol{x}_{t-1} \big) \neq 0    
\end{align}
where $\boldsymbol{x}_t = \big( x_{1t},...,  x_{pt} \big)^{\prime}$ and $p$ is a finite positive integer. 

In particular, the test above allows for more general cases that $y_t$ and $\boldsymbol{x}_t$ are of potentially different orders of integration. In particular, when both $y_t$ and $\boldsymbol{x}_t$ are stationary, then these are of the same integration order, there is no issue of unbalanced regression and testing the above hypothesis is equivalent to testing the coefficient $\boldsymbol{\beta} = 0$ in the linear regression model $y_t = \beta \boldsymbol{x}_{t-1} + \varepsilon_t$.  

\medskip

\begin{corollary}
Under the Assumption that both $\boldsymbol{x}_t$ and $y_t$ are stationary, we have that
\begin{align}
W_0 \equiv n \left(  \widehat{\boldsymbol{\beta}}_{ols} - \boldsymbol{\beta} \right)  \left[ \widehat{\mathsf{Avar} \left(  \widehat{\boldsymbol{\beta}}_{ols}  \right) } \right]^{-1} \left(  \widehat{\boldsymbol{\beta}}_{ols} - \boldsymbol{\beta} \right) \overset{d}{\to} \chi^2_p
\end{align}
where
\begin{align}
\widehat{\mathsf{Avar} \left(  \widehat{\boldsymbol{\beta}}_{ols}  \right) }  
= 
\left(  \frac{1}{n} \sum_{t=1}^n \boldsymbol{x}_{t-1} \boldsymbol{x}_{t-1}^{\prime} \right)^{-1}  \cdot \widehat{\mathsf{LRV} \big(  \boldsymbol{x}_{t-1} \varepsilon_t \big) } \cdot \left(  \frac{1}{n} \sum_{t=1}^n \boldsymbol{x}_{t-1} \boldsymbol{x}_{t-1}^{\prime} \right)^{-1}
\end{align}
where $\widehat{\mathsf{LRV} \big(  \boldsymbol{x}_{t-1} \varepsilon_t \big) }$ is the long-run variance of $\boldsymbol{x}_{t-1} \varepsilon_t$.
\end{corollary}

Our goal is to develop a test that accommodates both stationary and nonstationary $\boldsymbol{x}_t$. Therefore, we propose testing the condition $\mathsf{cov} ( y_t, \boldsymbol{x}_{t-1} ) = 0$ directly instead of testing $\beta_1 = 0$. However, because $\mathsf{cov} ( y_t, \boldsymbol{x}_{t-1} ) = 0$ depends on $t$ when $x_t$ is $I(1)$, we based our test on the limiting covariance, which we define as $\mathsf{lim}_{t \to \infty} \mathsf{cov} ( y_t, \boldsymbol{x}_{t-1} )$. As $\boldsymbol{x}_t$ is $I(1)$ or local-to-unity, $\mathsf{cov} ( y_t, \boldsymbol{x}_{t-1} )$ is no longer a constant and depends on $t$. Moreover, define with 
\begin{align}
\boldsymbol{\lambda}_{ y, \Delta \boldsymbol{x}  } 
:= 
\underset{ t \to \infty }{ \mathsf{lim} } \  \mathsf{cov} \big( y_t, \boldsymbol{x}_{t-1} \big).  
\end{align}
In particular, if $\boldsymbol{x}$ is $I(1)$ and suppose that $\mathsf{lim}_{t \to \infty} \mathsf{cov} ( y_0, \boldsymbol{x}_{0} ) \to 0$, then it holds that
\begin{align}
\boldsymbol{\lambda}_{ y, \Delta \boldsymbol{x}  } 
:= 
\underset{ t \to \infty }{ \mathsf{lim} } \  \mathsf{cov} \big( y_t, \boldsymbol{x}_{t-1} \big)
= \sum_{h=1}^{\infty}  \mathsf{cov} \big( y_t, \Delta \boldsymbol{x}_{t-h} \big)
\end{align}
which is well-defined when $\sum_{ h=1 }^{ \infty } \left| \mathsf{cov} \big( y_t, \Delta \boldsymbol{x}_{it-h} \big) \right|$ for $i \in \left\{ 1,..., p \right\}$.

\end{wrap}

\newpage

\subsection{Appendix of the paper}

\subsubsection{Proof of Lemma 1}

Using the weak law of large numbers in Andrews (1988, Theorem 2) we have that 
\begin{align*}
\underset{ \pi_1 \leq \uptau \leq \pi_0  }{ \text{sup} } \left| \hat{\beta}_1( \uptau) - \beta_1  \right| 
= 
\underset{ \pi_1 \leq \uptau \leq \pi_0  }{ \text{sup} } \left|  \frac{ \displaystyle \sum_{ t =  1 }^{ \floor{ n \uptau } }  y_{t-1}  \epsilon_t  }{ \displaystyle \sum_{ t =  1 }^{ \floor{  n \uptau } }  y_{t-1}^2 } \right|
\leq 
\frac{ n }{ \displaystyle \sum_{ t =  1 }^{ \floor{  n \uptau_1 } } y_{t-1}^2  } \underset{ \uptau_1 \leq \uptau \leq \uptau_0  }{ \text{sup} } \left|  \frac{ \displaystyle  \sum_{ t =  1 }^{ \floor{  n \uptau_1 } } y_{t-1} \epsilon_t   }{ n } \right| = \mathcal{O}_p (1) o_p(1) = o_p(1). 
\end{align*}
Moreover, we have that 

\begin{align}
\hat{\beta}_2( \uptau) - \beta_1 = \frac{ \displaystyle \sum_{ t = \floor{  n \uptau_0 } + 1 }^{ n  } y_{t-1}^2   }{ \displaystyle \sum_{ t = \floor{  n \uptau } + 1 }^{ n  } y_{t-1}^2  } \left( \beta_2 - \beta_1 \right) + \frac{ \displaystyle \sum_{ t = \floor{  n \uptau } + 1 }^{ n } y_{t-1} \epsilon_t  }{ \displaystyle \sum_{ t = \floor{  n \uptau} + 1 }^{ n } y_{t-1}^2  }, 
\end{align}

\begin{proof}
We have that 
\begin{align}
\hat{\beta}_2 ( \uptau ) 
= 
\frac{ \displaystyle \sum_{t= \floor{n \uptau}  + 1}^{n} y_t y_{t-1} }{ \displaystyle \sum_{t= \floor{n \uptau}  + 1}^{n} y^2_{t-1} }  
\equiv 
\frac{ \displaystyle \sum_{t= \floor{n \uptau}  + 1}^{n} \bigg\{ \beta_1 y_{t-1} I \left\{ t \leq k_0 \right\} + \beta_2 y_{t-1} I \left\{ t > k_0 \right\} + \epsilon_t \bigg\} y_{t-1} }{ \displaystyle \sum_{t= \floor{n \uptau}  + 1}^{n} y^2_{t-1} }   
\end{align}

\begin{align}
\hat{\beta}_2 ( \uptau ) 
= 
\beta_1  \frac{ \displaystyle \sum_{ t = \floor{n \uptau_0 }  + 1}^{n} y^2_{t-1} }{ \displaystyle \sum_{t= \floor{n \uptau}  + 1}^{n} y^2_{t-1} }   + 
\frac{ \displaystyle \sum_{t= \floor{n \uptau}  + 1}^{n} \beta_2 y_{t-1} I \left\{ t > k_0 \right\} y_{t-1} }{ \displaystyle \sum_{t= \floor{n \uptau}  + 1}^{n} y^2_{t-1} }  +   \frac{ \displaystyle \sum_{ t = \floor{  n \uptau } + 1 }^{ n  } y_{t-1} \epsilon_t  }{ \displaystyle \sum_{ t = \floor{  n \uptau } + 1 }^{ n  } y_{t-1}^2  },
\end{align}
Therefore, 
\begin{align}
\hat{\beta}_2 ( \uptau ) 
&=
\beta_1  \frac{ \displaystyle \sum_{ t = \floor{n \uptau_0 }  + 1}^{n} y^2_{t-1} }{ \displaystyle \sum_{t= \floor{n \uptau}  + 1}^{n} y^2_{t-1} }  
+
\frac{ \displaystyle \sum_{t= \floor{n \uptau_0 }  + 1}^{n} y^2_{t-1} }{ \displaystyle \sum_{t= \floor{n \uptau}  + 1}^{n} y^2_{t-1} } \beta_2 
+
\frac{ \displaystyle \sum_{ t = \floor{  n \uptau } + 1 }^{ n  } y_{t-1} \epsilon_t  }{ \displaystyle \sum_{ t = \floor{  n \uptau } + 1 }^{ n } y_{t-1}^2  },
\end{align}
and

\newpage 

\begin{align}
\hat{\beta}_2 (\uptau ) - \beta_1
&=
\beta_1  \frac{ \displaystyle \sum_{ t = \floor{n \uptau_0 }  + 1}^{n} y^2_{t-1} }{ \displaystyle \sum_{t= \floor{n \uptau}  + 1}^{n} y^2_{t-1} }  
+
\frac{ \displaystyle \sum_{t= \floor{n \uptau_0 }  + 1}^{n} y^2_{t-1} }{ \displaystyle \sum_{t= \floor{n \uptau}  + 1}^{n} y^2_{t-1} } \beta_2 
+
\frac{ \displaystyle \sum_{ t = \floor{  n \uptau } + 1 }^{ n  } y_{t-1} \epsilon_t  }{ \displaystyle \sum_{ t = \floor{  n \uptau } + 1 }^{ n  } y_{t-1}^2  } - \beta_1 ,
\nonumber
\\
&= 
\beta_1 \left( \frac{ \displaystyle \sum_{ t = \floor{n \uptau_0 }  + 1}^{n} y^2_{t-1} }{ \displaystyle \sum_{t= \floor{n \uptau}  + 1}^{n} y^2_{t-1} }    - 1 \right)
+
\frac{ \displaystyle \sum_{t= \floor{n \uptau_0 }  + 1}^{n} y^2_{t-1} }{ \displaystyle \sum_{t= \floor{n \uptau}  + 1}^{n} y^2_{t-1} } \beta_2 
+
\frac{ \displaystyle \sum_{ t = \floor{  n \uptau } + 1 }^{ n  } y_{t-1} \epsilon_t  }{ \displaystyle \sum_{ t = \floor{  n \uptau } + 1 }^{ n } y_{t-1}^2  }
\nonumber
\\
&=
\frac{ \displaystyle \sum_{ t = \floor{n \uptau_0 }  + 1}^{n} y^2_{t-1} }{ \displaystyle \sum_{t= \floor{n \uptau}  + 1}^{n} y^2_{t-1} } \big( \beta_2 - \beta_1 \big) + \frac{ \displaystyle \sum_{ t = \floor{  n \uptau } + 1 }^{ n  } y_{t-1} \epsilon_t  }{ \displaystyle \sum_{ t = \floor{  n \uptau } + 1 }^{ n  } y_{t-1}^2  }
\end{align}
\end{proof}
Note that similarly, we also have that 
\begin{align}
\hat{\beta}_1 ( \uptau ) - \beta_1
=
\frac{ \displaystyle \sum_{ t = \floor{n \uptau_0 }  + 1}^{ \floor{ n \uptau } } y^2_{t-1} }{ \displaystyle \sum_{ t = 1}^{ \floor{n \uptau} } y^2_{t-1} } \big( \beta_2 - \beta_1 \big) + \frac{ \displaystyle \sum_{ t = 1 }^{ \floor{  n \uptau } } y_{t-1} \epsilon_t  }{ \displaystyle \sum_{ t = 1 }^{ \floor{  n \uptau } } y_{t-1}^2  }, 
\end{align}

Next consider adding and subtracting the follow term on the expressions above
\begin{align}
\left( \beta_2 - \beta_1 \right) \frac{ \displaystyle n \frac{(1 - \uptau_0) \sigma^2}{1 - \beta_2^2} }{ \displaystyle \sum_{ t = \floor{  n \uptau } + 1 }^{ n } y_{t-1}^2 },
\end{align}
We obtain 
\begin{align}
\left[ \hat{\beta}_2 ( \uptau ) - \beta_1 \right] 
+ 
\left( \beta_2 - \beta_1 \right) \frac{ \displaystyle n \frac{(1 - \uptau_0) \sigma^2}{1 - \beta_2^2} }{ \displaystyle \sum_{ t = \floor{  n \uptau } + 1 }^{ n } y_{t-1}^2 } 
- 
\left( \beta_2 - \beta_1 \right) \frac{ \displaystyle n \frac{(1 - \uptau_0) \sigma^2}{1 - \beta_2^2} }{ \displaystyle \sum_{ t = \floor{  n \uptau } + 1 }^{ n } y_{t-1}^2 }
\end{align}

\newpage

and
\begin{align*}
RHS 
&\equiv 
\frac{ \displaystyle \sum_{ t = \floor{n \uptau_0 }  + 1}^{n} y^2_{t-1} }{ \displaystyle \sum_{t= \floor{n \uptau}  + 1}^{n} y^2_{t-1} } \big( \beta_2 - \beta_1 \big)
+ 
\frac{ \displaystyle \sum_{ t = \floor{  n \uptau } + 1 }^{ n  } y_{t-1} \epsilon_t  }{ \displaystyle \sum_{ t = \floor{  n \uptau } + 1 }^{ n  } y_{t-1}^2  }
+ 
\left( \beta_2 - \beta_1 \right) \frac{ \displaystyle n \frac{(1 - \uptau_0) \sigma^2}{1 - \beta_2^2} }{ \displaystyle \sum_{ t = \floor{  n \uptau } + 1 }^{ n } y_{t-1}^2 } 
- 
\left( \beta_2 - \beta_1 \right) \frac{ \displaystyle n \frac{(1 - \uptau_0) \sigma^2}{1 - \beta_2^2} }{ \displaystyle \sum_{ t = \floor{  n \uptau } + 1 }^{ n } y_{t-1}^2 }
\\
\\
&=
\frac{ \displaystyle \sum_{ t = \floor{  n \uptau } + 1 }^{ n  } y_{t-1} \epsilon_t  }{ \displaystyle \sum_{ t = \floor{ n\pi } + 1 }^{ n  } y_{t-1}^2  } 
+ 
\big( \beta_2 - \beta_1 \big) \left( \frac{ \displaystyle \sum_{ t = \floor{n \uptau_0 }  + 1}^{n} y^2_{t-1} }{ \displaystyle \sum_{t= \floor{n \uptau}  + 1}^{n} y^2_{t-1} } -  \frac{ \displaystyle n \frac{(1 - \uptau_0) \sigma^2}{1 - \beta_2^2} }{ \displaystyle \sum_{ t = \floor{ n \uptau } + 1 }^{ n } y_{t-1}^2 } \right) + \left( \beta_2 - \beta_1 \right) \frac{ \displaystyle n \frac{(1 - \uptau_0) \sigma^2}{1 - \beta_2^2} }{ \displaystyle \sum_{ t = \floor{  n \uptau } + 1 }^{ n} y_{t-1}^2 }
\end{align*}

\subsubsection{Derivations of Equations 6 and 7}

The asymptotic behaviour of the residual sum of squares is as following 
\begin{align}
\label{sup1}
&\underset{ \uptau_1 \leq \uptau \leq \uptau_0 }{ \text{sup} } \left| \frac{1}{n} \ RSS_n (\pi) - \sigma^2 -  \frac{ ( \uptau_0 - \uptau ) ( 1 - \uptau_0 ) ( \beta_2 - \beta_1 )^2 \sigma^2 }{ ( \uptau_0 - \uptau ) ( 1 - \beta_2^2 ) + (1 - \uptau_0 ) ( 1 - \beta_1^2 )  }  \right| = o_p(1), 
\\
\nonumber
\\
\label{sup2}
&\underset{ \uptau_0 \leq \uptau \leq \uptau_2 }{ \text{sup} } \left| \frac{1}{n} \ RSS_n (\pi) - \sigma^2 - \frac{ \uptau_0 ( \uptau - \uptau_0 ) ( \beta_2 - \beta_1 )^2 \sigma^2 }{ \uptau_0 ( 1 - \beta_2^2 ) + ( \uptau - \uptau_0 ) ( 1 - \beta_1^2 )  }  \right| = o_p(1), 
\end{align}

\paragraph{Proof of Expression \eqref{sup1}:}
For $\uptau_1 \leq \uptau \leq \uptau_0$, we can write the residual sum of squares as: 
\begin{align}
RSS_n ( \uptau) 
&= 
\sum_{t=1}^{ \floor{ n \uptau } } \bigg( \epsilon_t - \left( \hat{\beta}_1 ( \uptau ) - \beta_1 \right) y_{t-1} \bigg)^2 
+ 
\sum_{t = \floor{ n \uptau } + 1 }^{ \floor{ n \uptau_0 } } \bigg( \epsilon_t - \left( \hat{\beta}_2 ( \uptau ) - \beta_1 \right) y_{t-1} \bigg)^2
\nonumber
\\
&+
\sum_{t = \floor{ n \uptau_0 } + 1 }^{ n } \bigg( \epsilon_t - \left( \hat{\beta}_2 ( \uptau ) - \beta_2 \right) y_{t-1} \bigg)^2
\nonumber
\\
\nonumber
\\
&=
\sum_{t=1}^{ n } \epsilon_t^2 
- \frac{ \displaystyle  \left( \sum_{t = 1 }^{  \floor{ n \uptau } } y_{t-1} \epsilon_t \right)^2 }{ \displaystyle \sum_{t = 1 }^{  \floor{ n \uptau } }  y_{t-1}^2 } 
- 2 \left( \hat{\beta}_2 (\uptau ) - \beta_1 \right) \sum_{t = \floor{ n \uptau } + 1 }^{ \floor{ n \uptau_0 } } y_{t-1} \epsilon_t
\nonumber
\\
&+
\left( \hat{\beta}_2 ( \uptau ) - \beta_1 \right)^2   \sum_{t = \floor{ n \uptau } + 1 }^{ \floor{ n \uptau_0 } } y_{t-1}^2 - 2 \left( \hat{\beta}_2 (\uptau) - \beta_2 \right) \sum_{t = \floor{ n \uptau_0 } + 1 }^{ \floor{ n \uptau } } y_{t-1} \epsilon_t
\nonumber
\\
&+ 
\left( \hat{\beta}_2 ( \uptau ) - \beta_2 \right)^2   \sum_{t = \floor{ n \uptau_0 } + 1 }^{ n } y_{t-1}^2.
\end{align}

\newpage

\paragraph{Proof of Expression \eqref{sup2}:} For $\uptau_0 < \uptau \leq \uptau_2$ , we can write the residual sum of squares as 
\begin{align}
RSS_n (\uptau) 
&= 
\sum_{t=1}^{ \floor{ n \uptau_0 } } \bigg( \epsilon_t - \left( \hat{\beta}_1 (\uptau ) - \beta_1 \right) y_{t-1} \bigg)^2 
+ 
\sum_{t = \floor{ n \uptau_0 } + 1 }^{ \floor{ n \uptau } } \bigg( \epsilon_t - \left( \hat{\beta}_1 (\uptau ) - \beta_2 \right) y_{t-1} \bigg)^2
\nonumber
\\
\nonumber
\\
&+
\sum_{t = \floor{ n \uptau } + 1 }^{ n } \bigg( \epsilon_t - \left( \hat{\beta}_2 ( \uptau ) - \beta_2 \right) y_{t-1} \bigg)^2
\nonumber
\\
\nonumber
\\
&=
\sum_{t=1}^{ n } \epsilon_t^2 - 2 \left( \hat{\beta}_1 ( \uptau ) - \beta_1 \right) \sum_{t=1}^{ \floor{ n \uptau_0 } } y_{t-1} \epsilon_t
+ \left( \hat{\beta}_1 ( \uptau ) - \beta_1 \right)^2  \sum_{t=1}^{ \floor{ n \uptau_0 } } y_{t-1}^2 
\nonumber
\\
\nonumber
\\
&- 
2 \left( \hat{\beta}_1 ( \uptau ) - \beta_2 \right) \sum_{t = \floor{ n \uptau_0 } + 1 }^{ \floor{ n \uptau } } y_{t-1} \epsilon_t + \left( \hat{\beta}_1 ( \uptau ) - \beta_2 \right)^2 \sum_{t = \floor{ n \uptau_0 } + 1 }^{ \floor{ n \uptau } } y_{t-1}^2 
\nonumber
\\
\nonumber
\\
&- 
2 \left( \hat{\beta}_2 ( \uptau ) - \beta_2 \right) \sum_{t = \floor{ n \uptau } + 1 }^{ n }  y_{t-1} \epsilon_t + \left( \hat{\beta}_2 ( \uptau ) - \beta_2 \right)^2 \sum_{t = \floor{ n \uptau } + 1 }^{ n } y_{t-1}^2. 
\end{align}
After simplification of the last two terms we obtain:
\begin{align}
RSS_n ( \uptau) 
&= 
\sum_{t=1}^{ n } \epsilon_t^2 - 2 \left( \hat{\beta}_1 ( \uptau  ) - \beta_1 \right) \sum_{t=1}^{ \floor{ n \uptau_0 } } y_{t-1} \epsilon_t
+ \left( \hat{\beta}_1 ( \uptau ) - \beta_1 \right)^2  \sum_{t=1}^{ \floor{ n \uptau_0 } } y_{t-1}^2 
\nonumber
\\
\nonumber
\\
&- 
2 \left( \hat{\beta}_1 ( \uptau ) - \beta_2 \right) \sum_{t = \floor{ n \uptau_0 } + 1 }^{ \floor{ n \uptau } } y_{t-1} \epsilon_t + \left( \hat{\beta}_1 ( \uptau ) - \beta_2 \right)^2 \sum_{t = \floor{ n \uptau_0 } + 1 }^{ \floor{ n \uptau } } y_{t-1}^2 
\nonumber
\\
\nonumber
\\
&- \frac{ \displaystyle  \left( \sum_{t = \floor{ n \uptau } + 1 }^{ n } y_{t-1} \epsilon_t \right)^2 }{ \displaystyle \sum_{t = \floor{ n \uptau } + 1 }^{ n }  y_{t-1}^2 }. 
\end{align}

\newpage 

\begin{wrap}

\section{R Code}

\begin{small}

\subsection{Simulate Data Step}

\begin{verbatim}
#########################################################################
### Function 1: Simulate data pair under the null hypothesis
#########################################################################
    
function_simulate_null <- function(N = N, beta0=beta0, beta1=beta1, c1=c1, rho=rho)
{# begin of function
  
  N  <- N
  p  <- 1
  
  beta0 <- beta0
  beta1 <- beta1
 
  c1  <- c1 
  rho <- rho

  mu.vector <- matrix(0, nrow = p+1, ncol = 1 ) 
  Sigma     <- matrix( 0, nrow = (p+1), ncol = (p+1) )
  
  sigma_uv  <- rho*sqrt( 0.25 )*sqrt( 0.75 )
  
  Sigma[1,1] <- 0.25
  Sigma[1,2] <- sigma_uv
  Sigma[2,1] <- sigma_uv
  Sigma[2,2] <- 0.75
  
  # generate random error sequence from Multivariate Normal Distribution
  innov.e <- rmvnorm( n = N, mean = mu.vector, sigma = Sigma ) 
  innov.u <- as.matrix( innov.e[ ,1] )
  innov.v <- as.matrix( innov.e[ ,2] )
  
  rn <- ( 1 - c1/N )
  x  <- matrix(0,N,p)
  
    for(t in 2:N) 
    {  
      x[t,1] <- rn*x[(t-1),1] + innov.v[t,1]
    }
  
  x.t     <- as.matrix( x[2:N, ] )
  x.lag   <- as.matrix( x[1:(N-1), ] )
  beta    <- as.matrix( c( beta0, beta1 ) )
  innov.u <- as.matrix( innov.u[2:N,1] )
  
  y.t <- beta[1,1] + beta[2,1]*as.matrix(x.lag[ ,1]) + innov.u
  
  simulated.data <- structure( list( y.t = y.t, x.t = x.t, x.lag = x.lag ) ) 
  return( simulated.data )
  
}# end of function
#########################################################################
\end{verbatim}

\end{small}

\newpage

\subsection{Wald IVX Statistic Estimation Step}

\begin{verbatim}
#########################################################################
### Function 2: Estimate Wald IVX statistic 
#########################################################################

sup_Wald_IVX_function <- function(Yt=Yt_sim, Xt=Xt_sim, Xlag=Xlag_sim, 
delta =delta, cz=cz, pi0 = pi0)
{# begin of function 
  
  # Insert data
  Yt   <- Yt_sim
  Xt   <- Xt_sim
  Xlag <- Xlag_sim
  
  Yt   <- as.matrix( Yt )
  Xt   <- as.matrix( Xt )
  Xlag <- as.matrix( Xlag )
  
  delta <- delta 
  cz    <- cz
  pi0   <- pi0
  
  ####################################################################
  ## IVX Estimation Step ## 
  ####################################################################
  # length size of the time series
  n <- NROW(Xlag)
  # number of predictors
  p <- NCOL(Xlag)
  h <- 1

  # Rz assigns a common mildly integrated root to all regressors 
  Rz     <- ( 1 - cz / ( n^delta ) ) 
  diffx  <- as.matrix( Xt - Xlag)
  z      <- matrix(0, n, p)
  z[1, ] <- diffx[1, ]
  
  for (i in 2:n) 
  {
    z[i, ] <- Rz * z[i - 1, ] + diffx[i, ]
  }
  
  Z  <- rbind(matrix(0, 1, p), z[1:(n - 1),  , drop = F])
  ## Use the above estimated Xt, Zt and Yt to construct the Wald IVX statistic
  sup_Wald_IVX <- 0
  
  # Step 1: Estimate OLS regression
  n <- NROW(Xlag)
  # number of predictors
  p <- NCOL(Xlag)
  
  # Step 2: Estimate the Wald OLS statistic
  pi_seq  <- as.matrix( seq(pi0, (1 - pi0), by = 0.01) )
  dim     <- NROW( pi_seq )
  
  ones            <- matrix( 1, nrow = n, ncol = 1 )
  Wald_IVX_vector <- matrix( 0, nrow = dim, ncol = 1 )
  k_seq           <- matrix( 0, nrow = dim, ncol = 1 )
  
  for ( j in 1:dim )
  {
    k_seq[ j,1 ] <- floor( pi_seq[j ,1] * n  )
  }
  
  for (s in 1:dim)
  {# begin of for-loop
    
    k <- k_seq[ s,1 ]
    
    #####################################
    # Step 1: Obtain the estimate of the variance of OLS regression 
    #####################################
    
    Xlag1 <- as.matrix( Xlag )
    for (i in (k+1): n)
    {
      Xlag1[i, ] <- 0 
    }
    
    Xlag1_tilde <- cbind( ones, Xlag1  )
    for (i in (k+1): n)
    {
      Xlag1_tilde[i, ] <- 0 
    }
    
    Xlag2 <- as.matrix( Xlag )
    for (i in 1: k)
    {
      Xlag2[i, ] <- 0 
    }
    
    Xlag2_tilde <- cbind( ones, Xlag2  )
    for (i in 1: k)
    {
      Xlag2_tilde[i, ] <- 0 
    }
    
    regressors <- cbind( Xlag1_tilde, Xlag2_tilde )
    regressors <- as.matrix( regressors )
    
    model_OLS  <- lm( Yt ~ regressors - 1 )
    # Bols       <- coefficients( model_OLS )
    # Bols       <- as.matrix( as.vector( Bols  ) )
    
    epsilon_hat     <- matrix( residuals( model_OLS ) )
    cov_epsilon_hat <- as.numeric( crossprod( epsilon_hat ) / n )
    
    #####################################
    # Step 2: Obtain the estimates of the IVX estimators
    #####################################
    # Estimate the Z1 and Z2 matrices corresponding 
    # to before and after the structural break
    Z1 <- as.matrix( Z )
    for (i in (k+1): n)
    {
      Z1[i, ] <- 0 
    }
    
    Z1_tilde <- cbind( ones, Z1  )
    for (i in (k+1): n)
    {
      Z1_tilde[i, ] <- 0 
    }
    
    Z2 <- as.matrix( Z )
    for (i in 1: k)
    {
      Z2[i, ] <- 0 
    }
    
    Z2_tilde <- cbind( ones, Z2  )
    for (i in 1: k)
    {
      Z2_tilde[i, ] <- 0 
    }
    
    B1_hat_ivx    <- ( t( Yt ) %*% (Z1_tilde) ) 
    %*% pracma::pinv( t(Xlag1_tilde) %*% Z1_tilde )
   
    B2_hat_ivx    <- ( t( Yt ) %*% (Z2_tilde) ) 
    %*% pracma::pinv( t(Xlag2_tilde) %*% Z2_tilde )
    
    Bivx_distance <- as.matrix( B1_hat_ivx  - B2_hat_ivx )
    
    Q_matrix_part1 <- ( pracma::pinv( t(Z1_tilde)%*%Xlag1_tilde )
    %*%( t(Z1_tilde)%*%Z1_tilde ) %*%pracma::pinv( t(Z1_tilde)%*%Xlag1_tilde ) )
    
    Q_matrix_part2 <- ( pracma::pinv( t(Z2_tilde)%*%Xlag2_tilde )
    %*%( t(Z2_tilde)%*%Z2_tilde )%*%pracma::pinv( t(Z2_tilde)%*%Xlag2_tilde ) )
    
    Q_matrix           <- ( Q_matrix_part1 + Q_matrix_part2 )
    Wald_IVX_statistic <- as.numeric( ( 1 / cov_epsilon_hat ) * ( Bivx_distance )
    %*%(pracma::pinv( Q_matrix )) %*%t( Bivx_distance ) )
    
    Wald_IVX_vector[s,1] <- Wald_IVX_statistic 
    
  }# end of for-loop
  # Obtain the sup-Wald statistic
  sup_Wald_IVX <- max( Wald_IVX_vector  )
  
  return( sup_Wald_IVX )  
}# end of function
#####################################################################



\end{verbatim}

\end{small}

\end{wrap}

\newpage 
   
\bibliographystyle{apalike}
\bibliography{myreferences1}

\end{document}